\documentclass[twoside, 12pt]{article}
\usepackage{pgfplots,comment}
\usepackage[a4paper,top=1.18in, bottom=1.18in, left=1.2in, right=1.2in,bindingoffset=0mm]{geometry}

\usepackage{aditya2}

\tikzset{
    >=latex,
    pil/.style={
            draw,
      <-, 
      decorate,
      decoration={snake,,amplitude=.02cm, pre length=.2cm,post length=.2cm,}
              }}

\hypersetup{colorlinks=true,linkcolor=red, citecolor= blue, urlcolor=darkred}

\definecolor{BluishGreen}{RGB}{0,158,115}

\title{{\bf \sc {{\color{darkblue}
Similarity of Information and Collective Action}}}\footnote{This research was supported in part by National Science Foundation award 2417694. Aditya Kuvalekar acknowledges British Academy Small Grant SRG2223\textbackslash231506. We are grateful to Nageeb Ali, Marco Battaglini, Steve Callander, Mehmet Ekmekci, Albin Erlanson, Daniel Garrett, Elliot Lipnowski, John Maxwell, Meg Meyer, Emre Ozdenoren, Jacopo Perego, Doron Ravid, Phil Reny, Vasiliki Skreta, Yu Fu Wong, Jidong Zhou, and conference audiences and seminar participants at Bocconi University, CERGE-EI, University of Chicago, Collegio Alberto, Essex University, HKUST, Kelley School of Business, Northwestern University, Oxford University, Penn State University, Princeton University, Stony Brook Game Theory meeting, SITE Political Economy, Toulouse School of Economics, University College London, University of Warwick, and Yale University for their helpful
comments and suggestions. Basak: Indiana University, Kelley School of Business, email: dbasak@iu.edu; Deb: New York University, email: joyee.deb@nyu.edu, Kuvalekar: University of Essex, email:  a.kuvalekar@essex.ac.uk.} \\}

\author{                   
\begin{minipage}{0.3\textwidth}\centering 
Deepal Basak\\ \centering \it \small Indiana University
\end{minipage}                  
\begin{minipage}{0.3\textwidth}\centering 
Joyee Deb   \\ \centering \it \small New York University
\end{minipage}                  
\begin{minipage}{0.3\textwidth}\centering 
Aditya Kuvalekar
\\ \centering \it \small University of Essex
\end{minipage}   
}

\bigskip

\date{\today}

\begin{document}

\maketitle
\bigskip
\begin{center}

\end{center}
\bigskip
\vspace{-1cm}
\begin{abstract}
\noindent We study a canonical collective action game with incomplete information. Individuals attempt to coordinate to achieve a shared goal, while also facing a temptation to free-ride. Consuming more similar information about the fundamentals can help them coordinate, but it can also exacerbate free-riding. Our main result shows that more similar information facilitates (impedes) achieving a common goal when achieving the goal is sufficiently challenging (easy). We apply this insight to show why insufficiently powerful authoritarian governments may face larger protests when attempting to restrict press freedom, and why informational diversity in committees is beneficial when each vote carries more weight.

\bigskip

\vspace{0.2cm}

\end{abstract}
\thispagestyle{empty}
\newpage
\pagenumbering{arabic}

\begin{spacing}{1.05}

\section{Introduction}

This paper addresses the following question: When does increased similarity of information among participants help or harm participation in collective action? 

A collective action problem is a situation 
in which individuals want to achieve a common goal  but face a temptation to free-ride (see \cite{tullock1971paradox}, \cite{olson2012logic}), because reaching the goal requires a sufficient number of people to take a costly action, while the benefit accrues to all.\footnote{Examples of collective action problems are ubiquitous: protests, regime changes, boycotts, or voting in committees. See \cite{palfrey1985voter}, \cite{taylor1982chickens}, \cite{goeree2005explanation}, \cite{myatt2008does}, \cite{diermeier2008voting}, \cite{shadmehr2018protest}, \cite{dziuda2021difficulty} for some examples.} We consider collective action problems under incomplete information, in which individuals face uncertainty about the benefits of reaching the goal, and can privately learn about it. In such situations, an individual may not take the costly action even after learning that reaching the goal is socially beneficial. This is because her decision depends not only on what her private information tells her about the state (fundamental uncertainty), but also on what it tells her about the information that others have and thus what they will do (strategic uncertainty). 
For instance, an individual's private information may make her believe, at the same time, that reaching the goal is beneficial, but that others do not intend to take the costly action, rendering her own costly action ineffective. We  
investigate what happens to participation in  collective action when information becomes more similar, in the sense that agents know that others have private information with content similar to their own. 

Understanding the effects of changing similarity of information in strategic environments is particularly important against the backdrop of the extraordinary changes in how information is disseminated and consumed in recent years. Algorithms steer individuals to news or video content based on personal characteristics, resulting in like-minded people accessing the same information from the same source. A plausible conjecture is that if people with the same objectives now access the same  information, it may be easier for them to predict each other's actions and attain better outcomes.

Our central  observation is that greater similarity of information among agents, \emph{even those with identical preferences}, can serve as a double-edged sword in collective action games. On the one hand, if people believe that others are more likely to have the same information as them, they may be able to coordinate better to reach the common goal. On the other hand, the temptation to free-ride may be exacerbated: if an agent knows that others have the same information and predicts that they will take action, then  she does not need to take a costly action herself. 
Both these opposing effects have been observed separately in empirical work on protests, an archetypal collective action game. \cite{enikolopov2020social} find, using Russian data, that in cities where individuals accessed news from the same media platform, it led to more protests, and cities in which people were not all on the same platform faced fewer protests. However, in an experiment on  Hong Kong protests, \cite{cantoni2019protests} show essentially the opposite:  knowledge of others’ participation led to a stronger temptation to free-ride for potential participants. These opposite effects motivate our research question: when is information similarity helpful and when is it harmful in collective action games? Our main results characterize when increased information similarity helps or harms participation in a canonical regime-change game of incomplete information. To illustrate the main forces, let us start with a simple example.

\noindent \textbf{Example:} 
Abe and Bob are working on a project that is of high ($\stater=1$) or mediocre ($\stater=0$) quality with equal probability. They do not know the project quality, but each receives a private binary signal about it. For simplicity, we assume that a player never mistakes a mediocre project for a high-quality one but may mistake a high-quality project for a mediocre one. Formally, each player~$i$ receives a binary signal $\sigr_i$, such that  $\sigr_i = 0$ if $\stater = 0$, and if $\stater = 1$, the signals are drawn from some exchangeable joint distribution $\jcdf^1$. Importantly, the signals might not be conditionally independent.

After observing the private signals, the two players decide (independently and  simultaneously) whether to work or shirk. Working entails a cost of $c$. The project is completed with certainty if both players work  and completed with probability $q$ if only one works. If a high-quality project is completed, the players enjoy a benefit of $1$ each. If the project is mediocre or unfinished, they both get $0$. Table~\ref{Table: game payoff intro example} specifies the payoffs. For simplicity, suppose also that $c$ is high enough that players will never work after receiving signal $\sigr_i = 0$. 
\begin{center}
\begin{table}[ht]
    \centering
    \begin{tabular}{cc|c|c|}
            & \multicolumn{1}{c}{} &                     \multicolumn{2}{c}{Bob}\\
            & \multicolumn{1}{c}{} & \multicolumn{1}{c}{work} & \multicolumn{1}{c}{shirk} \\ 
            \cline{3-4}
   \multirow{2}{*}{Abe}     & work &$\theta-c,\theta -c$& $q\theta -c, q \theta$ \\ 
            \cline{3-4}
    & shirk & $q\theta,q\theta - c$ &$0,0$ \\ 
   \cline{3-4}
\end{tabular}
    \caption{Payoff matrix}
    \label{Table: game payoff intro example}
\end{table}
    \end{center}
\vspace*{-6mm}

In this example, notice that while both players enjoy the benefit from the completion of a high-quality project, each player may have an  incentive to free-ride  even if he believes the project is high quality. We ask first whether there is a symmetric, pure-strategy equilibrium~$\strategy^1$ in which a player works whenever $\sigr_i = 1$---that is, when he knows the project is of high quality. Then we ask, what happens if  the realized private signals that Abe and Bob observe become more similar, in the sense that each player knows that his private signal is now a better predictor of the other player's signal (and his consequent action choice).

For a player to be willing to work after observing $\sigr_i = 1$, we need \begin{align*}
    (1-q) \jcdf^1_1(\sigr_{-i} = 1 ) + q \jcdf^1_1(\sigr_{-i} = 0) \ge \cost.
\end{align*}
where $\jcdf^1_1$ is the conditional distribution of $\sigr_{-i}$ given that $\sigr_i = 1$.\footnote{Due to exchangeability of $\jcdf^1$, this conditional distribution need not be indexed by $i$.} 
Intuitively, conditional on observing $\sigr_i = 1$, player~$i$ works only if he believes he is pivotal (his effort makes a difference) with high-enough probability. Recall that if a player works, then his marginal contribution to the completion of the project is $(1-q)$ if the other player also works and $q$ if the other player does not work.  
The marginal benefit from effort depends on two primitives: (i) the extent to which individual effort can make a difference, measured by $q$, and (ii) a player's conditional belief about his opponent's signal (and implied equilibrium action). It is easy to see that there are $c$ and $\jcdf^1$ (when $\jcdf^1_1(\sigr_{-i}=1)<\frac12$) such that the following result holds: 
\begin{quote}
    There exists $q^* \in (0,1)$ such that for all $q \in [q^*,1]$, there is an equilibrium in which each player~$i$ works whenever $\sigr_i = 1$.\footnote{We provide a complete analysis of this example along with the parametric configurations supporting these results in Appendix~\ref{Appendix: intro example}.} 
\end{quote}
\noindent What happens when Abe's and Bob's signals become \emph{more similar?}, i.e., it becomes more likely that Abe and Bob see the same signals?  Formally, suppose that when the project is high quality ($\stater = 1$), the probability that both receive the same (different) signal increases (decreases) by $\alpha$. They continue to receive a signal $0$ when $\stater = 0$. Moreover, suppose that, for any given agent, the probability that he receives a signal $1$ conditional on $\stater =1$ remains unchanged. Figure~\ref{Figure: CAD increase, binary, intro} illustrates such a change in the joint distribution in state~$\stater=1$. 
\begin{figure}[h!]
    \centering
\begin{tikzpicture}
    \filldraw[blue] (0,0) circle (4pt);
    \filldraw[red] (0,2) circle (4pt);
    \filldraw[blue] (2,2) circle (4pt);
    \filldraw[red] (2,0) circle (4pt);
    \node at (0,-0.4) {  $0$};
    \node at (2,-0.4) {  $1$};
    \node at (-0.4,0) {  $0$};
    \node at (-0.4,2) {  $1$};
    \node at (0.35,0) {  $+\alpha$};
    \node at (0.35,2) {  $-\alpha$};
    \node at (2.35,0) {  $-\alpha$};
    \node at (2.35,2) {  $+\alpha$};
    \node at (1,-0.8) {  $X_1$};
    \node at (-0.8,1) {  $X_2$};
    \draw[thick,color=black,->] (0,2) -- (0,0.05);
    \draw[thick,color=black,->] (2,0) -- (2,1.95);
\end{tikzpicture}
    \caption{Increasing similarity when $\stater = 1$}
    \label{Figure: CAD increase, binary, intro}
\end{figure}

How does this increased similarity in the realized private signals affect players' incentives to work in equilibrium? The answer to this question depends on the value of $q$. A player works on a signal $1$ if 
\begin{align*}
    (1-q) \underbrace{\jcdf^1_1(\sigr_{-i} = 1 )}_{\uparrow \text{ with more similarity}} + q \underbrace{\jcdf^1_1(\sigr_{-i} = 0)}_{\downarrow \text{ with more similarity}} \ge \cost.
\end{align*}
When $q$ is high, with high likelihood  one player's effort suffices for project completion. In this case, after observing $\sigr_{i}=1$, a player~$i$ 
is less likely to be pivotal compared to before and has a stronger incentive to free-ride. When $q$ is low, after observing $\sigr_{i}=1$, a player~$i$ assigns a higher probability to being pivotal  and has a stronger incentive to work. Indeed, there exists an $\alpha$ which makes $\jcdf^1_1(\sigr_{-i}=1)>\frac12$ such that the following holds:
\begin{quote}
There exists $q^{**} \in (0,1)$ such that, with more similar information, there is an equilibrium in which each player~$i$ works after observing $\sigr_{i} = 1$ if and only if $q \in [0, q^{**}]$.
\end{quote}

\noindent That is, after observing $\sigr_i = 1$, under less similar signals, agents are willing to work if and only if $q$ is sufficiently high. With more similar information, the exact opposite is true: agents are willing to work if and only if $q$ is sufficiently low. Intuitively, a low value of $q$ captures environments in which not achieving coordination is the primary obstacle to collective action. In these settings, increasing information similarity proves beneficial. A high value of $q$ captures environments in which free-riding is the primary obstacle to collective action. In these settings, increasing information similarity may be harmful. 

Of course this example has several simplifying features. With binary signals, it was straightforward to define what it means to have more similar information. Moreover, the parameter~$q$ was a mechanical way of quantifying the relative importance of coordination (versus preventing free-riding). Finally, strategies are quite simple with binary signals and given no agent participates after a signal $0$. 
But it turns out that the essential insights from this example generalize. 

In the baseline model, we consider a canonical regime-change game. There are two states of the world: one in which regime change is beneficial for society, and one in which it is not.  There are two geographically dispersed groups of uncertain size. The groups have identical preferences  but have access to different information about the state. Each group receives a signal about the state, drawn from a finite set. Each individual sees only the signal of their own group and decides whether to participate in an action aimed to bring about regime change. Regime change is successful only if a (potentially random) threshold number of agents participate. The benefit of a successful regime change is public, but the cost of participation is borne only by the participants. In this canonical setting, we ask whether increased similarity of groups' signals increases or decreases participation. 

First, we need a notion that allows us to compare similarity between information structures. We use an order of information similarity, analogous to the one in the  example above. We say information is \emph{more similar} or has a \emph{higher concentration along diagonal} (CAD) when, conditional on observing signal realization~$x$, an agent believes it is more likely that others also observed the signal realization~$x$ and less likely that others observed a signal realization different from $x$.\footnote{For two-dimensional signals, our CAD order is equivalent to one proposed by \cite{meyer1990interdependence}. Like her, we require that the two signals have the same marginal distribution over states. } 

We then show that, analogous to low and high values of~$q$ in the  example, it is possible to cleanly partition the general setting into \emph{encouragement environments}, in which greater aggregate participation encourages participation by making it more likely that an individual can make a difference, and \emph{discouragement environments}, in which greater aggregate participation discourages participation by making it less likely that an agent can make a difference. 

Our main results establish that more similar information in the CAD order  increases participation in  encouragement environments but can reduce participation in discouragement environments. Formally, for any strategy profile, we define the participation (nonparticipation) set as the set of signals for which agents participate (do not participate). Given any information structure, our game can admit multiple equilibria, including, possibly equilibria with strategies that are non-monotonic in signals. 
We prove that more similar information enlarges the participation set in the equilibrium with maximal participation in encouragement environments (Theorem~\ref{Theorem: CAD helps in encouragement}) but can shrink it 
in discouragement environments (Theorem~\ref{Theorem: CAD hurts in discouragement}).

We also characterize these two environments. Under an intuitive single-crossing condition, it turns out that  resilient regimes, in which regime change requires a large number of participants, constitute  encouragement environments and weak regimes, in which regime change requires a small number of participants, constitute discouragement environments. So  the effect of changing information similarity is qualitatively
different based on the resilience of regimes. Recall the mixed empirical evidence regarding how modern information technology has affected protests.  Our results imply that increased similarity of information may have facilitated greater participation in collective action against regimes previously thought to be impregnable but, at the same time, hindered movements with ex-ante easier goals.

We apply our framework to two applications. First, we study mass protests \cite[as in][]{shadmehr2018protest} and investigate the effect of increased information similarity on the likelihood of a successful protest and consequent welfare. Notice that even when increased similarity of information increases participation, it may not improve welfare or the likelihood of success. This is because increased similarity of information increases the probability  that everyone participates and  the probability that no one participates, making the  overall effect on the probability of  success unclear a priori. Welfare effects are also ambiguous because even if increased participation makes regime change more likely, the marginal social benefit may be lower than the marginal cost. Specializing to an environment with a deterministic participation threshold and Poisson uncertainty about group sizes we show that increasing similarity can never lower the maximal probability of a successful protest or welfare in encouragement environments but can do so in discouragement environments. These findings have important implications for press freedom and authoritarian governments. When authoritarian governments curb press freedoms, they effectively reduce the similarity of information across individuals. Our results suggest that it is precisely the resilient authoritarian governments that benefit more from doing this. Less powerful regimes may face larger protests and a higher likelihood of being overthrown when attempting to restrict press freedom. 

Second, we apply our framework to a setting of costly voting in a committee \cite[as in][]{palfrey1985voter}. We consider a monetary policy committee with $n>2$ members who must vote on whether to raise interest rates. A rate increase is implemented only if a threshold number of members vote for it. Public opinion is against any rate increases, and committee member votes are observed ex post, making it costly for any member to vote in favor of a rate increase. Committee members base their decision on their private information about whether a rate increase is warranted based on the state of the economy. Their information can be more or less similar  depending on their backgrounds---academic backgrounds or areas of expertise that makes them focus on different aspects of the available evidence. 
How do changes in member diversity affect the voting outcome? In this setting, there is no uncertainty about the number of participants or the vote threshold. And importantly, we need to move beyond bivariate signals. We develop a natural analog of our CAD order to compare multivariate binary signals, and use it to establish results analogous to Theorems~\ref{Theorem: CAD helps in encouragement} and \ref{Theorem: CAD hurts in discouragement}. We show that conditional on the
rate increase being warranted, a more diverse committee (with less similar information) strengthens the incentive to vote correctly  and increases the
maximal equilibrium number of votes in favor of a rate increase  if each individual
vote carries enough weight---that is, fewer votes are required for a rate hike. Conversely, if an individual vote carries very little weight, a more diverse committee weakens the incentive to vote correctly.\footnote{Unlike Theorems \ref{Theorem: CAD helps in encouragement} and \ref{Theorem: CAD hurts in discouragement}, however, with more than two groups, the characterization is only partial---that is, not every environment can be classified as encouragement or discouragement. We provide more details in Section~\ref{Section: robustness}.} We also examine how changes in information similarity influence the optimal voting threshold rule.

We conclude by examining some other extensions of our baseline model.

\medskip

\subsection{Related Literature}\label{sec:relatedliterature}

A large literature going back to at least \cite{H71}  studies how exogenous change in the information environment changes agents' incentives in strategic environments. More recently, \cite{morris2002social}, \cite{angeletos2007efficient}, \cite{bergemann2013robust}, \cite{jensen2018distributional}, and \cite{mekonnen2022bayesian} have studied this question in a class of games with monotone best responses (pure complementarity or substitutability). Monotonic best responses imply that extremal equilibria are in monotone strategies (see, for example, \cite{van2007monotone}). This literature provides insights about how changes in the exogenous information structure affect the monotone equilibria and welfare. Since a canonical collective action game involves coordination and free-riding motives, the best response is nonmonotonic, making the established tools unsuitable. Moreover, much of the literature cited above focuses on the effect of new public information. Importantly, we study how the equilibrium set changes with \emph{changes in similarity of information}.

We consider an arbitrary signal structure and propose  CAD as a natural order of information similarity. 
The existing literature contains other measures of the interdependence of joint distributions (for example,  \cite{muller2002comparison}, \cite{meyer2012increasing}), but none are appropriate for comparing the \emph{conditional}  belief distributions that arise in strategic settings with incomplete information. For the bivariate case (the focus of this paper), the CAD order is the same as that in \cite{meyer1990interdependence}. 
Like her, we consider multivariate random variables with fixed marginal distributions while changing the joint distribution. 
\cite{clemen1985limits} and \cite{CB24} study how such changes impact the value of information and show that informational diversity may be valuable. \cite{de2023robust} consider an environment with known marginal distributions but unknown joint distribution to obtain the robustly optimal policy in a class of decision problems. \cite{awaya2022common} study the effect of the interdependence of signals on common learning as do 
\cite{cripps2008common}. They show that essentially any interdependence obstructs common learning.
 
We also contribute to the sizable literature on protests and voting. The majority of the theoretical work on protest focuses on the coordination aspect.  However, some recent papers---for example, \cite{shadmehr2018protest}, \cite{dziuda2021difficulty}, and \cite{parkglobal2022}---incorporate free-riding and construct cutoff equilibria. \cite{dziuda2021difficulty} show that 
a lower required participation threshold for success might not increase the likelihood of a successful protest. \cite{mutluer2024lead} shows that a higher cost of participation may lead to larger protest. These papers do not consider changing information environments, which is our focus. Some recent empirical papers have studied the effect of modern communication technologies on the size of protests. For instance, \cite{manacorda2020liberation} show empirical evidence that mobile phones facilitated protests in Africa, and \cite{enikolopov2020social} show that the diffusion of an online social network increased protest turnout in Russia. However, \cite{cantoni2019protests} demonstrate, in a recent experiment about   mass protesters in Hong Kong, that the knowledge of others' participation led to a stronger temptation to free-ride for potential participants. These ambiguous empirical results underscore the importance of our research question. We provide a clear characterization of when information similarity helps and when it hurts. Our results have implications for the effect of press freedom on protests. \cite{edmond2013information} considers a game of protest in which the regime can manipulate information. However, in this game, the agents only have the coordination motive and no free-riding motive. In our voting application, we study how informational diversity affects voting incentives. \cite{taylor2010public} and \cite{roesler2022committee} study similar questions. In these papers, if no other agent votes, then an agent can always get her desired outcome by voting. In contrast, in our setup, if others do not vote, then an agent alone can never get her desired outcome. \cite{chemmanur:2018} present a model of voting on a corporate board that closely resembles our application. In a recent paper, \cite{kattwinkel2023} characterize the optimal voting rule under conditionally independent signals. We restrict attention to simple threshold voting rules but show how information similarity affects the optimal voting rule.

\section{A Regime-Change Game}\label{sec:regimechangegame}
There are two states of the world: $\stater\in \states= \{0,1\}$. In $\stater=1$, it is socially beneficial to change the regime. In $\stater=0$, it is not socially beneficial to change the regime. An alternative interpretation of $\stater$ is that it captures whether a regime change is feasible ($\stater=1$) or not ($\stater = 0$). Society consists of $\ngroups$ groups. For most of the paper, we consider $G=2$. The analysis for more than two groups is relegated to the online appendix.
We introduce population uncertainty \`a la \cite{myerson1998population}: the number of agents in any group~$g$, denoted by $\popr_g$, is a $\mathbb Z_+$-valued random variable with probability mass function $\ngpmf(\cdot)$ and mean $\pop$. Agents do not observe the size of their own group or other groups. We let $\ngpmfag(\cdot)$ denote the conditional probability mass function of $\popr_g-1$ according to an agent in group $g$.\footnote{
In general, $\ngpmfag(\cdot)$ can be different from $\ngpmf(\cdot)$ because an agent may be more or less optimistic about the size of  her group, conditional on belonging to the group herself, compared to what someone outside the group believes about the group size. As an example, if 
 $\ngpmf(\cdot)$ were drawn from a Poisson distribution, then $\ngpmfag(\cdot)$ and $\ngpmf(\cdot)$ would coincide.   } 
 
Each agent decides whether to take a costly action. Participating (choosing $a = 1$) costs $c > 0$, while not participating ($a=0$) is costless. Regime change occurs only if enough agents participate. We assume that the threshold participation required to change the regime is an $\mathbb N$-valued random variable, $\thresholdr+1$, where $\thresholdr$ follows  probability mass function $\tpmf(\cdot)$, and  agents do not observe the realization of $\thresholdr+1$. We call $\thresholdr$ the \emph{resilience} of the regime.  We allow population uncertainty and a random (unobserved) resilience because this is more realistic in many regime-change settings. However, in Section~\ref{Section: voting} we present an application with observable and deterministic group sizes and participation threshold. We summarize the payoffs in the matrix below, in which $\threshold+1$ denotes the realized resilience and $\Action$ denotes the number of agents who participate.
\medskip
\begin{center}
\begin{tabular}{ |c|c|c| } 
 \hline
    & $\Action \ge \threshold +1$ & $\Action \le \threshold$ \\ 
    \hline
 $a=1$ & $\state-c$ & $-c$\\ 
 \hline
 $a= 0$ & $\state$ & $0$ \\ 
 \hline 
\end{tabular}
\end{center}

Note that regime change when it is not beneficial~(in $\stater = 0$) entails no additional costs beyond the costs of participation. This assumption can be relaxed (see  Section~\ref{Section: robustness}). Finally, $\stater,\popr_1,\popr_2,\thresholdr$ are independent.

\subsection{Information Structure}\label{sec:information}
Before deciding whether to participate, agents receive information about the state of the world. There is a fixed, finite set of signals~$\signalset$. Each group receives a signal~$\sigr_g$ drawn from~$\signalset$, and every agent in group~$g$ observes only the signal received by their own group. Let $\sigr:= (\sigr_g)_{g\in\groupset}$. 
We denote the joint distribution of $(\stater,\sigr)$ by $\prob(\cdot) \in \Delta(\states \times \signalset^\ngroups)$, and the distribution of $\sigr$ conditional on $\stater = \state$ by $\jcdf^\state \in \De(\signalset^\ngroups)$. As we will describe in Section~\ref{Section: similarity}, similarity of information in our context will simply be a measure of interdependence of $\sigr$. 

We assume that $\sigr_1,\sigr_2$ are independent of $\popr_1,\popr_2$ and $\thresholdr$. That is, a group's  signal conveys information about the state to the agents but not about their or the other group's realized size or the threshold. Let $\prior:= \prob(\{\stater = 1\})$. We denote by $\margdiststate_g \in \De(\signalset)$ the distribution of $\sigr_g$ given that $\stater=\state$. We assume that the distribution of $\jcdf^\state(\cdot)$ is exchangeable so that we have $\margdiststate_g = \margdiststate_{g'} =: \margdiststate$ for all $g,g'\in \groupset$ and all $\state \in \states$.

Let $\mu(\signal):= \prob(\{\stater = 1\} \vert \{\sigr_g=\signal\})$ denote the posterior probability that any agent in group $g$ assigns to the state's being $1$ given a realized signal $\signal$. We assume that $\mu: \signalset \to [0,1]$ is injective. 
Given group~$g$, we let $\sigr_{-g}$ 
be the random variable denoting the signal of the other group. 
Let $\ccdf^\state_\signal\in \Delta (\signalset)$ denote the conditional distribution of $\sigr_{-g}$ 
given state $\stater=\state$ and $\sigr_g = \signal$. Since $\margdiststate$ is exchangeable, $\ccdf^\state_\signal(\cdot)$ is the same for every group.

\subsection{Strategies and Aggregate Participation}
\textbf{Strategies:} A (pure) strategy of agent $i$ from group $g$ is a mapping, $$\strategy_g:\signalset\to \{0,1\}.$$ 
That is, we assume symmetric strategies within a group. We restrict attention to pure strategies throughout the paper, until Section~\ref{Section: robustness}, in which we discuss the extension to mixed strategies. Given a strategy profile $\strategy = (\strategy_1,\strategy_2)$, we define the \emph{participation set of $\strategy_g$ for group $g$}, denoted by $P(\strategy_g)$, to be the set of signals such that $\strategy_g(\signal)= 1$. 
Analogously, we define the \emph{nonparticipation set of $\strategy$}, denoted by $NP(\strategy_g):= \signalset\backslash P(\strategy_g)$. When the dependence on $\strategy_g$ is obvious, we denote $P(\strategy_g)$ and $NP(\strategy_g)$ by $P_g$ and $NP_g$, respectively. 
\medskip
\newline
\textbf{Aggregate participation:} Since group size is random, we let $\Actionr$ denote the random variable corresponding to the total number of participating agents, given a strategy profile. 
$$\Actionr := \sum_{g=1}^{\ngroups} \popr_g \strategy_g(\sigr_g)$$
We call $\Actionr$ the aggregate participation. Notice that $\Actionr$ is $(\popr_1,\popr_2,\sigr)$-measurable and depends on $\strategy$. When this is obvious, we suppress $\strategy$. 
All agents in  group $g$ receive the same signal and have the same belief. Let $\Actionr_{-g}(\signal_g,\signal_{-g};\strategy)$ be the aggregate participation according to an agent in group $g$, excluding herself. Then we have 
\begin{eqnarray*}
\Actionr_{-g}(\signal_g,\signal_{-g};\strategy) &:=& \popr_{-g} \strategy_{-g}(x_{-g}) + (\popr_g-1) \strategy_g(x_g)\\ 
&=&\popr_{-g} \ind_{\sig_{-g} \in P_{-g}} + (\popr_g-1) \ind_{\sig_g \in P_g}.
\end{eqnarray*}
\noindent
\textbf{Expected aggregate participation in $\stater=1$: }

For any $\strategy
$, define 
\begin{align}
    \protsize(\strategy):= \E[\Actionr(\sigr;\strategy)\vert \stater = 1]
\end{align}
to be the expected aggregate participation in state $1$, which is the state in which it is beneficial to change the regime. For a fixed $\strategy$, $\protsize(\strategy)$ depends only on the marginal distribution of the signals and not on the joint distribution. So information similarity affects $\protsize$ only by affecting the equilibrium $\strategy$. We state this in the following lemma. The proof is in the appendix. With some abuse of notation, let $\mcdf^1(S)$ denote  $\prob(\sigr_g \in S \vert \stater = 1)$.

\blemma\label{Lemma: Expected protest size} For any $\strategy$, with associated participation sets $(P_g)_{g \in \groupset}$, 
$$\protsize(\strategy) = \pop\sum_{g=1}^\ngroups \mcdf^1(P_g).
$$
Therefore, when $\ngroups = 2$, we can say 
$\protsize(\strategy) = 2[ \mcdf^1(P_1 \cup P_2) + \mcdf^1(P_1 \cap P_2)]$.
\elemma
 \noindent
\textbf{Solution concept:} 
We consider Bayes Nash equilibria in pure strategies. We do not impose any additional structure on the equilibrium, such as monotone or symmetric strategies. Multiple equilibria may exist, including one in which no one participates regardless of the signal. 

Notice that since we have a collective action problem, best responses are not monotonic in aggregate participation, unlike in other regime-change games with only strategic complementarities \cite[][for example]{morris2002social}. Therefore, it is not clear whether equilibria can be ordered in any natural way. This means we cannot use existing tools---such as those used in supermodular games---directly. Let $\eq(\jcdf)$ be the set of strategy profiles that constitute an equilibrium under information structure $\jcdf$. 

Given multiple equilibria, we focus on how increased similarity affects the maximal possible participation in any equilibrium. Accordingly, we define the following.
\bdefn[\textbf{Maximal Participation Equilibrium and Maximal Equilibrium Aggregate Participation}] We say that an equilibrium $\strategy^*$ is a maximal participation equilibrium 
if $\protsize(\strategy^*) \ge \protsize(\strategy)$ for all $\strategy \in \eq(\jcdf)$.\footnote{Since the set of signals is finite and we look at pure strategies, the existence of a maximal aggregate participation equilibrium is guaranteed. } 
Let $\protsize^*(\jcdf)$ denote the expected aggregate participation (in state~1) in the maximal participation equilibrium given information structure $\jcdf$, and call it the maximal equilibrium aggregate participation. 
\edefn 

Both the expected aggregate participation given a strategy~~$\protsize(\strategy)$ and the maximal equilibrium aggregate participation for an information structure $\protsize^*(\jcdf)$ are defined \emph{conditional on~$\stater=1$}; that is,  when change is beneficial. For brevity, henceforth, we do not mention this explicitly. $\protsize(\cdot), \protsize^*(\cdot)$ also depend on other parameters, such as $\thresholdr$. We typically suppress this dependence and only make the dependence on the information structure explicit.\footnote{Most of our results about the expected aggregate participation remain unchanged if we used the ex-ante expected aggregate participation rather than expected aggregate participation when $\stater=1$.} 

Arguably, instead of studying maximal aggregate participation, we could have focused on the maximal equilibrium probability of successful collective action or maximal welfare. One reason to focus on participation is its empirical relevance. Recent work in political science on collective action measures aggregate participation in mass protests.\footnote{For instance, several empirical studies \cite[][e.g.]{enikolopov2020social}, as well as popular press articles (e.g.,  \url{https://www.nytimes.com/2023/11/22/opinion/does-protest-work-bevins.html}) and even datasets that document protests over time \cite[][e.g.]{DVN/HTTWYL_2016} record turnout.} 
In Section~\ref{sec:resilience of regimes}, we study also the effect of information similarity on the probability of successful regime change and on welfare (in the canonical case where $\popr_1,\popr_2$ are Poisson distributed), and show that results analogous to our main results continue to hold.

\subsection{A Measure of Information Similarity}\label{Section: similarity}
Given our research question, we need a notion of informational similarity. Agents make participation decisions based on their beliefs about the state of the world and the expected aggregate participation. So they must reason about the \emph{conditional} probability of others' information given their own. We use the following similarity order for two-dimensional random variables using such conditional beliefs.

\bdefn[\textbf{Concentration Along Diagonal (CAD)}]\label{Definition: CAD} Let $\mathcal Y \subset \real$ be a finite set. Let $Y$ and $\widehat Y$ be two $\mathcal Y^2$-valued exchangeable
random variables whose distributions are given by $\ydist$ and $\yhatdist$, respectively. We say $Y$ is more similar than $\widehat Y$ in the CAD order, denoted by $Y\cad \widehat Y$ or $\ydist \cad \yhatdist$, if the following two conditions hold.
\begin{enumerate}
    \item $Y_i$ and $\widehat Y_i$ are identically distributed for all $i\in\{1,2\}$.
    \item For $y \in \mathcal Y$ and $T \subseteq \mathcal Y$,  
    \begin{enumerate}
        \item $\ydist(Y_2 \in T\vert Y_1 = y) \ge \yhatdist(\widehat{Y_2} \in T \vert \widehat{Y_1} = y)$ if $y \in T$.  
        \item $\ydist(Y_2 \in T\vert Y_1 = y) \le \yhatdist(\widehat Y_2 \in T \vert \widehat Y_1 = y)$ if $y \notin T$.
    \end{enumerate}
\end{enumerate}
\edefn 

Notice that, by exchangeability of the distributions, we can interchange $Y_1$ and $Y_2$ in the definition. We  use the CAD order to compare $\ccdf^\state_\signal(\cdot)$, the beliefs of players conditional on a state, and a realized signal, keeping the marginal beliefs conditional on $\stater = \state$ by $\margdiststate$ unchanged.\footnote{We analyze the effect of increasing similarity of $\ccdf^1_\signal(\cdot)$. This is because, given the definitions of states and payoffs in our game, changing $\ccdf^0_\signal(\cdot)$ is not payoff relevant if marginals~$\margdiststate$ are unaltered. See Section~\ref{Section: robustness} for a more detailed discussion.}

The CAD order captures the idea that when information becomes more similar, any agent believes that it is now more likely that others received the same signal as they did. Notice that agents face two types of uncertainty in our environment: fundamental uncertainty about $\stater$, and strategic uncertainty about the other group's information. Part~1 in the definition means that an increase in CAD keeps the fundamental uncertainty unchanged, and potentially varies only the strategic uncertainty. 

The large literature that studies the value of \emph{public} information in coordination games \cite[for example,][]{morris2002social} alters both fundamental and strategic uncertainty at once since public information changes both, the joint and the marginal distribution at once. 
Our approach of exploring the implications of varying joint distributions while keeping the marginal distributions fixed does arise in some existing papers. \cite{clemen1985limits} study the value of information in a class of decision problems where the noises are jointly normally distributed. Like us, they fix the marginal distribution and vary the correlation structure. \cite{CB24} consider a general information structure with fixed marginals while varying correlations conditional on the state.\footnote{\cite{de2023robust} also study a robust decision problem with known marginals and unknown correlation that is similar in spirit.} 
By keeping the fundamental uncertainty unchanged, our formulation isolates the effect of strategic uncertainty. In Section \ref{Section: robustness} we demonstrate how our results are robust to allowing the marginal distributions to also change. 

To gain more intuition about the notion of CAD, it is useful to consider a simple numerical example of two information structures that are CAD-ordered. Suppose~$(\widehat Y_1,\widehat Y_2)$ is an information structure with $\widehat Y_i$'s being conditionally independent. Now consider a new information structure $Y$, with the same marginal distribution as $\widehat Y$. But now the $Y_i$'s are conditionally independent with probability~$1-\varepsilon$ for some $\varepsilon>0$ and perfectly correlated with probability~$\varepsilon$. $Y$ is more similar than $\widehat Y$ in the CAD order. To see how such a CAD increase in information similarity can arise in practice, consider how people have converged on where they get information from. 
For instance, with YouTube being the dominant player and algorithms steering content to users, it becomes more likely (higher $\varepsilon$) that people now view exactly the same content. This corresponds to a CAD increase.

An equivalent way of stating Definition~\ref{Definition: CAD} is that $Y\cad \widehat Y$ if the probability of events in which $\widehat Y_1 ,\widehat Y_2$ are \emph{exactly equal} must (weakly) increase under $Y$, and the probability of events in which $\widehat Y_1 ,\widehat Y_2$ are unequal must (weakly) decrease under $Y$. This alternate formulation is equivalent to an order by 
\cite{meyer1990interdependence}. The requirement that the two agents see exactly the same signal with a higher probability may be too strong in some contexts, making the CAD order quite incomplete.  Indeed, two information structures~$Y$ and $\widehat Y$ are not comparable in the CAD order if the probability of events in which  $\widehat Y_1 ,\widehat Y_2$ are \emph{very close in value} increased under $Y$ and the probability of events in which $\widehat Y_1 ,\widehat Y_2$ are \emph{significantly different} decreased.
We discuss the implications of using alternative, more complete orders after presenting our main results.

Finally, notice that we construct posteriors explicitly using the signals, rather than model signals as posteriors themselves, as is now standard, following \cite{Kamenica2011}. Our assumption that $\posterior(\cdot)$ is injective implies that signals and posteriors are interchangeable. We could have chosen signals as posteriors and started with a \emph{feasible} joint distribution over posterior beliefs instead. With more than one agent, characterizing the feasible joint distributions over posteriors is not trivial. Recently, \cite{arieli2021feasible} characterize the set of feasible two-dimensional joint distributions.  By working with signals directly and performing the CAD operations, we have feasible joint distributions by construction.

\section{Information Similarity and Participation}
In this section, we present our main results that  characterize how increased similarity of information affects participation in equilibrium. 
\subsection{Preliminaries}

In an equilibrium of our regime-change game, an agent is willing to bear the cost of participation if and only if she believes that it is sufficiently likely that a change is beneficial and that her participation will make a difference. 

Consider an agent in group $g \in \{1,2\}$. She believes that the size of the other group is $\popr_{-g}\sim \eta(.)$. However, that she belongs  to group $g$ may change her belief about the size of her own group. She believes that $\popr_{g}-1\sim \ngpmfag(\cdot)$. 
Let $\totalnpmfag(\cdot)$ denote her belief about $\popr_{g}-1+\popr_{-g}$.\footnote{Since $\popr_1$ and $\popr_2$ are i.i.d., this distribution does not need to be indexed by group $g$.} The agent's participation incentive depends on whether she expects her own participation to make a difference to the protest's outcome. Below, we define  expressions~$\pivtwo$, $\pivone$, and $\pivother$, which denote, respectively, the probabilities of an agent being pivotal when both groups participate, when only her own group participates, and when only the other group participates.
\begin{equation}
    \pivtwo := \sum_{k=0}^\infty \tpmf(k) \totalnpmfag(k), \ 
    \pivone := \sum_{k=0}^\infty \tpmf(k) \ngpmfag(k), \ 
    \pivother:=  \sum_{k=0}^\infty \tpmf(k) \ngpmf(k) \label{Equation: pivotal probabilities}
\end{equation}

In collective action games, agents take a costly action when they believe  they can make a difference (are pivotal). While models with pivotality are ubiquitous in the voting literature, one common critique is that voters in large electorates are unlikely to ever be pivotal. But in our setting with uncertainty about the population size and the required threshold, $\pivtwo$, $\pivone$, and $\pivother$ are smooth functions of pivotal probabilities and reflect agents' beliefs about whether their participation can make a difference. 
So, $\pivtwo$, $\pivone$, and $\pivother$ serve as a way of modeling this incentive to avoid costly participation in any collective action game with strategic uncertainty.

We assume that  from any agent's perspective, the probability of being pivotal when only the other group participates is weakly smaller than the probability of being pivotal when only her own group participates or when both  groups participate.\footnote{This assumption is satisfied if, for example, $\popr_1,\popr_2$ are Poisson random variables, as in, for example,  \cite{myerson1998population}. In contrast, if $\popr_1,\popr_2$ are deterministic, say equal to $\pop$, then Assumption~\ref{Assumption: pivotal with one or two groups more likely than pivotal other} is violated if $Prob(\thresholdr = \pop + 1) > \max\{Prob(\thresholdr = \pop), Prob(\thresholdr = 2\pop\}$.  Assumption~\ref{Assumption: pivotal with one or two groups more likely than pivotal other} rules out cases in which an agent from a group that does not participate has the strongest incentive to participate,  knowing that no one in his group will participate.} We maintain the assumption below throughout the paper. 

\bass\label{Assumption: pivotal with one or two groups more likely than pivotal other}$\max\{\pivone,\pivtwo\} \ge \pivother$ \eass

\noindent
We first write down the conditions for a strategy profile to be an equilibrium.
\bprop
\label{prop:equilibrium}
A strategy profile $\strategy$ is an equilibrium if and only if for all $g \in \{1,2\}$ and for all $\signal \in \signalset$, 
\begin{align}
    \jcdf^1_\signal(P_{-g})\pivtwo +(1-\jcdf^1_\signal(P_{-g}))\pivone\ge \frac{c}{\mu(\signal)} & \text{ if } \signal \in P_g \tag{IC:P}\label{Equation: IC P}\\
    \jcdf^1_\signal(P_{-g}) \pivother \le \frac{c}{\mu(\signal)} & \text{ if } \signal \in NP_{g}. \tag{IC:NP}\label{Equation: IC NP}
\end{align}
\eprop

The intuition is straightforward.  
Consider an agent in group $g$ with signal $\signal \in P_g$. If the other group also receives a signal in its  participation set, which occurs with probability $\jcdf^1_\signal(P_{-g})$, then this agent can make a difference with probability $\pivtwo$. If the other group does not receive a signal in $P_{-g}$, then the agent can make a difference with probability $\pivone$. \eqref{Equation: IC P} simply says that the agent has an incentive to  incur the cost of participating if she believes she can make a difference with a sufficiently high probability. 
The logic behind \eqref{Equation: IC NP} is similar.

To capture how the incentives change with similarity of information, it is convenient to partition the model primitives into two environments.\footnote{In the knife-edge case with $\Lambda_b=\Lambda_o$, changing information similarity has no impact on participation.}

\begin{definition}\label{Definition: encouragement and discouragement}[Encouragement/Discouragement]
{\color{white}move to next line}
\begin{itemize}
\item 
We say we are in an \textbf{encouragement environment} if
\begin{equation}
    \enco.
    \tag{E}
    \label{enco}
\end{equation}
In this case, an agent is more likely to make a difference when both groups participate than if only her own group participates; that is, higher aggregate participation encourages participation.  
\item 

We say we are in a \textbf{discouragement environment} if 
\begin{equation}
    \disco.
    \tag{D}
    \label{disco}
\end{equation} 
In this case,  an agent is more likely to make a difference when only her group participates than if both groups participate; that is, higher participation by others discourages individual participation. 
\end{itemize}
\end{definition}
At first glance, encouragement and discouragement environments may seem to be environments of strategic complementarity and substitutability, respectively. This is not quite true. Discouragement environments do not feature strategic substitutability, because a nonparticipating agent in group $g$ (with  $x_g\in NP_g$) has a \emph{stronger} incentive to participate if the other group is more likely to participate. 

Theorems~\ref{Theorem: CAD helps in encouragement} and \ref{Theorem: CAD hurts in discouragement} establish that if we compare information similarity in the sense of CAD, then the above simple condition about primitives (whether $\pivtwo > \pivone$ or $\pivone > \pivtwo$) yields a complete characterization of when increased information similarity facilitates or  hinders participation.

\subsection{Encouragement Environment}\label{Section: Encouragement Environment}
\bthm\label{Theorem: CAD helps in encouragement}
In encouragement environments, the maximal equilibrium aggregate participation increases when information becomes more similar. That is,  
\newline
\[\jcdf^{\theta} \cad \widehat{\jcdf}^{\theta} \ \text{for all} \ \theta \implies  \protsize^*(\jcdf) \ge \protsize^*(\widehat \jcdf) \ \text{in} \ \eqref{enco}.  
\]
This is true regardless of how $\jcdf^{0}$ changes, as long as the marginals $\margdiststate$ are unaltered.
\ethm
The proof is in the appendix and proceeds in two steps. First, we show that in encouragement environments, the maximal participation equilibrium must be in symmetric strategies. Then we show that any maximal participation equilibrium remains an equilibrium when information similarity increases. Suppose that in the maximal equilibrium, in each group, an agent participates if and only if $\signal \in P$. If information similarity increases, an agent with $\signal\in P$ now assigns a higher probability that the other group also sees a signal $\signal\in P$ that induces them to participate. That is, $\jcdf^1_\signal(P)$ increases for $\signal\in P$. Since $\enco$, we can see from \eqref{Equation: IC P} that such an agent has an even stronger incentive to participate. Analogously, if information similarity increases, a nonparticipating agent with $\signal\in NP$ now assigns a lower probability that the other group sees $\signal\in P$. That is, $\jcdf^1_\signal(P)$ decreases for $\signal \in NP$. We can see from \eqref{Equation: IC NP} that a nonparticipating agent has an even weaker incentive to participate. 

\subsection{Discouragement Environment}\label{Section: Discouragement Environment}

Next, we analyze discouragement environments, 
in which $\disco$. The maximal equilibrium might no longer be symmetric. Further, a symmetric equilibrium under $\widehat \jcdf$ might no longer be an equilibrium under $\jcdf$ when $\jcdf \cad \widehat \jcdf$. This alone does not imply a smaller maximal equilibrium participation under $\jcdf$, because new equilibria may arise under $\jcdf$ that were not sustainable under $\widehat \jcdf$. Given any information structure~$\jcdf$ and a maximal equilibrium~$\strategy^*$, we define  a condition that describes why $\strategy^*$ is maximal under $\jcdf$.

\bdefn[Condition M]\label{Definition: maximality due to free-riding} Let $\strategy^*$ be a maximal equilibrium for $\jcdf$ with participation sets $(P_1^*,P_2^*)$. We say that $\jcdf$ satisfies condition M if, for any 
strategy profile~$\widehat \strategy$ with participation sets~$(\widehat P_1,\widehat P_2)$ 
 such that $\protsize(\widehat \strategy)> \protsize(\strategy^*)$, 

at least one of the following holds. 
\begin{enumerate}[(i)]
    \item[(M1)] $\exists \signal \in \widehat P_1 \cap \widehat P_2$ such that 
\begin{align*}
    \min_{i \in \{1,2\}} \left\{\jcdf^1_\signal(\widehat P_i) \pivtwo + (1-\jcdf^1_\signal(\widehat P_i)) \pivone\right\} < \frac{c}{\posterior(\signal)}.
\end{align*}
Or 
\item[(M2)] $\exists \signal \in (\widehat P_1 \cup \widehat P_2) \backslash (\widehat P_1 \cap \widehat P_2)$ such that
\begin{align*}
    \ind_{\signal \in \widehat P_1} \jcdf^1_\signal\left(\widehat P_1\right) + \ind_{\signal \in \widehat P_2} \jcdf^1_\signal\left(\widehat P_2\right)  > \frac{c}{\pivother \posterior(\signal)}.
\end{align*}
\end{enumerate}
\edefn
Condition M says there are two reasons why any strategy profile~$\widehat \strategy$ with a larger expected participation than the maximal equilibrium~$\strategy^*$ fails to be an equilibrium. 
Either \eqref{Equation: IC P} is violated for 
some signal that prescribes both groups to participate under 
$\widehat \strategy$, or  
\eqref{Equation: IC NP} is violated at a signal  at which exactly one group is prescribed to participate under $\widehat \strategy$. Below we establish that in discouragement environments, if the information structure satisfies condition~M, increasing information similarity can lead to lower maximal equilibrium aggregate participation.

Condition M is not a condition on the primitives of the model. However, a straightforward sufficient condition that guarantees condition M is as follows. Let $\signalset$ be ordered according to the posterior beliefs, $\posterior(\signal)$. Define $\widehat\signal:= \inf \{ \signal: \posterior(\signal) \pivone \ge \cost\}$. If $\widehat\strategy := \ind_{\sigr \ge \widehat \signal}$ is an equilibrium for some $\jcdf$, then $\jcdf$ satisfies condition M. The reason is that  regardless of the information structure, \eqref{Equation: IC P} can never be satisfied for $\signal < \widehat \signal$. Therefore, no equilibrium can have a larger expected participation than $\widehat\strategy$. And therefore condition M is satisfied.  
One easily verifies with a binary-signal example that this sufficient condition is not vacuous. Equipped with this condition, we now present our second main result, which  establishes how increased information similarity can reduce  maximal aggregate participation.

\bthm\label{Theorem: CAD hurts in discouragement} In a discouragement environment that satisfies condition M, the maximal equilibrium participation decreases when information becomes more similar. That is,  
\newline
\[\jcdf^\theta \cad \jcdfhat^\theta, \ \text{for all} \ \theta \implies \protsize^*(\jcdf) \le \protsize^*(\jcdfhat) \ \text{under} \ \eqref{disco} \ \text{if} \ \jcdfhat \ \text{satisfies condition} \ M.  
\]
Moreover, the inequality can be strict. The result is true regardless of how $\jcdf^{0}$ changes, as long as the marginals $\margdiststate$ are unaltered.
\ethm 
The proof is in the appendix. The argument involves two steps. First, we argue that the maximal equilibrium might no longer be an equilibrium when information becomes more similar. Let $\strategy^*$ with participation sets~$(P_1^*,P_2^*)$ be the maximal equilibrium under $\jcdfhat$. 
Consider a participating agent with a signal $\signal\in P_1^*\cap P_2^*$. If information becomes more similar, this agent assigns a higher probability to the event that the other group also receives a signal in their respective participation set. 
However, unlike in encouragement environments, this reduces her incentive to participate since $\disco$. As a result, this agent's \eqref{Equation: IC P} may be violated. Indeed, a nonparticipant's \eqref{Equation: IC NP} may also fail. Consider a signal $\signal \in P_1^* \backslash P_2^*$. 
An agent in group $2$ who receives such a signal is prescribed to not participate. However, with increased similarity of information, this agent assigns a higher probability that group $1$ will participate. 
This, in turn, makes her more likely to participate, which may violate \eqref{Equation: IC NP}. So the maximal equilibrium $\strategy^*$ under $\jcdfhat$  may no longer be an equilibrium under $\jcdf$. 

In the second step, we establish that no new equilibrium with larger expected participation arises under $\jcdf$. Suppose, for a contradiction, there is an equilibrium~$\strategy'$ with $\protsize(\sigma')> \protsize(\sigma^*)$. By the maximality of $\sigma^*$, we know that $\sigma'$ is not an equilibrium under $\jcdfhat$. By condition~M, two cases arise. In case (i), $\sigma'$ is not an equilibrium under $\jcdfhat$ because an agent's incentive to participate \eqref{Equation: IC P} is violated at some signal~$\signal\in P_1' \cap P_2'$, where both groups are prescribed to participate. With more similar information, \eqref{Equation: IC P} would continue to be violated. To see why, note that \eqref{Equation: IC P} is a convex combination of $\pivtwo$ and $\pivone$, and with more similar information, she assigns a higher weight to $\pivtwo$. Now if \eqref{Equation: IC P} was violated under $\jcdfhat$, then it will also be violated under $\jcdf$ because in discouragement environments, $\pivtwo\le \pivone$. In case (ii), $\sigma'$ is not an equilibrium under $\jcdfhat$ because \eqref{Equation: IC NP} is violated for some agent---say, from group $2$---with a signal in $(P_1'\backslash P_2')$. Such an agent wishes to participate under $\jcdfhat$ because she assigns a high probability to the event that group $1$ participates. 
Under $\jcdf$, when information is more similar, she has an even stronger incentive to participate (since she believes that the other group is more likely to participate and her own group is not going to participate). So, in both cases, if $\strategy'$ is not an equilibrium under $\jcdfhat$, then it cannot be an equilibrium under~$\jcdf$ either.

\noindent\textbf{Discussion of Theorems~\ref{Theorem: CAD helps in encouragement} and \ref{Theorem: CAD hurts in discouragement}:} A few observations are worth highlighting. First, the existing literature often restricts attention to monotone, cutoff strategies for the sake of tractability. 
In a game of pure complementarity \cite[see][]{morris2002social}, the best and worst equilibria are in cutoff strategies. However, with both complementarities and substitutabilities, 
this need no longer be true. We do not impose such a restriction on strategies; we allow nonmonotonic and asymmetric strategies. 
Also, we only assume a fixed finite signal space without imposing any additional restrictions on the signal structure. Consequently, there is little hope of characterizing all equilibria. However, using an indirect approach, we can characterize the effects of changing similarity of information on the set of equilibria and the maximal ones. Also note that we restrict attention to pure strategies. In encouragement environments, this is without loss: maximal equilibrium is in symmetric pure strategies. In discouragement environments, this need not be true. However, the qualitative insight that increased information similarity can lead to lower participation in discouragement environments remains valid even if the maximal equilibrium is in mixed strategies. See Section~\ref{Section: robustness} for a detailed discussion.

Second,  we started by partitioning the underlying primitives into the encouragement and discouragement environments  to capture when coordination and free-riding, respectively, are the primary hurdles to collective action. This partition was independent of any information structure. Theorems~\ref{Theorem: CAD helps in encouragement} and \ref{Theorem: CAD hurts in discouragement} imply that measuring information similarity using the notion of CAD yields the intuitive economic insight that increasing information similarity helps collective action exactly when coordination is the main challenge  and can hinder it  when free-riding is the main challenge.

This raises the natural question  whether our characterization can be obtained using an order of similarity that is less demanding (more complete) than CAD. Given information structures $Y$ and $\widehat Y$ with the same marginals, we say $Y\ccad \widehat Y$ if 
$$Prob(Y_2 \in A \vert Y_1 =x) \ge Prob(\widehat Y_2 \in A\vert \widehat Y_1 =x) $$
for all $x$ and all \emph{upper- and lower-contour sets $A$ containing $x$}.\footnote{A set $A$ is an upper-contour set if $A = \{ z' \in \signalset: z' \ge z\}$ for some $z \in \signalset$. A similar definition applies for a lower-contour set.} Intuitively, increased similarity no longer means that there is a higher chance of getting \emph{exactly the same} signal. Rather, $Y\ccad \widehat Y$ if $Y_1$, $Y_2$ are \emph{close to each other in value} with a higher probability relative to $\widehat Y_1, \widehat Y_2$.  We can obtain a result similar to Theorem~\ref{Theorem: CAD helps in encouragement} using this weaker order (under a mild regularity condition on the information structure). 
The maximal equilibrium then has a cutoff structure in encouragement environments:  agents participate if they receive a signal above a threshold. An increase in similarity even in this weaker order still implies an increase in participation in encouragement environments. However, in discouragement environments, maximal equilibria might not be in cutoff strategies, and, in general, an analogous result to Theorem~\ref{Theorem: CAD hurts in discouragement} does not hold. The extant literature often allows only cutoff strategies---for example, \cite{shadmehr2018protest}. If we consider only cutoff strategies, then Theorem~\ref{Theorem: CAD hurts in discouragement} holds with the weaker order. In this sense, CAD, even though demanding, enables a characterization of the basic trade-off in the paper without imposing restrictions on equilibrium strategies.  

Similarly, one might aim to characterize changing information similarity using more familiar orders like the supermodular order or the PQD order \cite[see][for details]{meyer2012increasing}. In two dimensions, the CAD order is stronger than the supermodular order, which is equivalent to the PQD order. As previously mentioned, encouragement and discouragement environments divide the game into scenarios where strategic complementarities or strategic substitutabilities prevail. It is straightforward to create examples where increasing similarity in the supermodular order (and hence the PQD order) results in lower participation in encouragement environments. Thus, the CAD order provides a characterization that aligns with our initial intuition---that increasing similarity aids coordination and exacerbates free-riding---something that orders like the supermodular order do not seem to provide.

\subsection{Encouragement, Discouragement, and Resilience}

Roughly speaking, encouragement and discouragement environments are the general analogs of the parameter~$q$ in our  example in the introduction. They quantify the relative importance of coordination (versus preventing free-riding). Recall that the definitions of these environments (Definition~\ref{Definition: encouragement and discouragement}), while based on primitives, were in terms of the population uncertainty and the threshold uncertainty, making them somewhat abstract. Below we provide an alternate characterization that shows that intuitively, a higher threshold for regime change makes an encouragement and a lower threshold makes a discouragement environment.  

We say a regime with  participation threshold~$\thresholdr'$ is more resilient than one with threshold~$\thresholdr$, denoted by $\thresholdr'\succcurlyeq_{st} \thresholdr$, if $\thresholdr'$ first-order stochastically dominates $\thresholdr$.

\begin{proposition}\label{Proposition: resilience and environments} 
Suppose that a single-crossing assumption holds---that is, for all $\thresholdr'\succcurlyeq_{st} \thresholdr$, if $(\pivtwo-\pivone)(\thresholdr)\ge0$ then $(\pivtwo - \pivone)(\thresholdr')\geq 0$. 
Define $\thresholdr$ such that $(\Lambda_2-\Lambda_1)(\thresholdr)=0$.
\begin{enumerate}
\item For any $\thresholdr'$ with $\thresholdr'\succcurlyeq_{st} \thresholdr$, we are in an encouragement environment. 
\item For any $\thresholdr'$ with $\thresholdr\succcurlyeq_{st} \thresholdr'$, we are in a discouragement environment.
\end{enumerate}
\end{proposition}
The single-crossing assumption means that if an agent is more likely to make a difference when both groups join than when only her group joins, this continues to be true when regime change becomes more difficult---that is, if $\pivtwo>\pivone$ for a given  regime, then $\pivtwo>\pivone$ also for more resilient regimes. 
The assumption is satisfied, for example, if the group sizes~$\popr_1,\popr_2$ are drawn from a Poisson distribution and the resilience $\thresholdr$ is deterministic. In this case, 
$\enco$ if and only if $\threshold> n^*=\frac{\pop}{\ln 2}$. Figure~\ref{fig:enco and disco poisson} illustrates encouragement and discouragement environments for the Poisson example. 

\begin{figure}[h!]
\bigskip
\centering
    \begin{subfigure}[b]{0.45\textwidth}
        \centering
        \begin{tikzpicture}
            \begin{axis}[
                width=8cm, height=6cm,
                no markers,
                axis lines=middle,
                xlabel={$x$},
                ylabel={$\ppdf(\bar n, x)$},
                xlabel style={at={(axis description cs:1,0)}, anchor=north},
                ylabel style={at={(axis description cs:0.15,.75)}, anchor=south},
                xtick=\empty,
                ytick=\empty,
                axis line style={-latex},
                every axis label/.append style={font=\large},
                domain=0:20,
                samples=100,
                ymin=-0.025,
                ymax=0.2
            ]
                \addplot[ultra thick, blue] {exp(-x) * x^12 / factorial(12)};
            \addplot[only marks, mark=*, mark options={scale=1.2, fill=black}] coordinates {(7.5, 0.03657544193551429)};
            
            \node at (axis cs:7.5,-0.015) {$N$}; 
            
            \addplot[dashed, thick, black] coordinates {(7.5, 0.03657544193551429) (7.5,0)};

            \addplot[only marks, mark=*, mark options={scale=1.2, fill=black}] coordinates {(12, 0.11436791550944653)};
            
            \node at (axis cs:12,-0.015) {$\threshold$}; 
            
            \addplot[dashed, thick, black] coordinates {(12, 0.11436791550944653) (12,0)};

            \addplot[only marks, mark=*, mark options={scale=1.2, fill=black}] coordinates {(15, 0.08285923436864542)};
            
            \node at (axis cs:15,-0.015) {$2N$}; 
            
            \addplot[dashed, thick, black] coordinates {(15, 0.08285923436864542) (15,0)};
            \end{axis}
        \end{tikzpicture}
        \caption{Encouragement (high $\threshold$)}
        \label{fig:encouragement}
    \end{subfigure}
    \hfill
    \begin{subfigure}[b]{0.45\textwidth}
        \centering
        \begin{tikzpicture}
            \begin{axis}[
                width=8cm, height=6cm,
                no markers,
                axis lines=middle,
                xlabel={$x$},
                ylabel={$\ppdf(\bar n, x)$},
                xlabel style={at={(axis description cs:1,0)}, anchor=north},
                ylabel style={at={(axis description cs:0.15,.75)}, anchor=south},
                xtick=\empty,
                ytick=\empty,
                axis line style={-latex},
                every axis label/.append style={font=\large},
                domain=0:20,
                samples=100,
                ymin=-0.025,
                ymax=0.2
            ]
                \addplot[ultra thick, blue] {exp(-x) * x^7 / factorial(7)};

                \addplot[only marks, mark=*, mark options={scale=1.2, fill=black}] coordinates {(5.5, 0.12344930223720431)};
            
            \node at (axis cs:5.5,-0.015) {$N$}; 
            
            \addplot[dashed, thick, black] coordinates {(5.5,0.12344930223720431) (5.5,0)};

            \addplot[only marks, mark=*, mark options={scale=1.2, fill=black}] coordinates {(7, 0.1490027796743379)};
            
            \node at (axis cs:7,-0.015) {$\threshold$}; 
            
            \addplot[dashed, thick, black] coordinates {(7, 0.1490027796743379) (7,0)};

            \addplot[only marks, mark=*, mark options={scale=1.2, fill=black}] coordinates {(11, 0.06457716255760958)};
            
            \node at (axis cs:11,-0.015) {$2N$}; 
            
            \addplot[dashed, thick, black] coordinates {(11,0.06457716255760958) (11,0)};
            \end{axis}
        \end{tikzpicture}
        \caption{Discouragement (low $\threshold$)}
        \label{fig:discouragement}
    \end{subfigure}
    \caption{\small{Encouragement and discouragement environments. Here, $\ppdf(\threshold,x) = \frac{e^{-x} x^{\threshold}}{\threshold!}$ is the Poisson pmf at $\threshold$ with mean $x$. The peak of these pmfs occurs at $\threshold.$ In both the environments, $\pop < \threshold < 2\pop$. However, in encouragement environments, $\threshold$ is much larger than $\pop$. Therefore, $\pivone = \ppdf(\threshold,\pop) < \ppdf(\threshold,2\pop) = \pivtwo$ in encouragement environments (left), while the reverse holds in discouragement environments (right).}}
    \label{fig:enco and disco poisson}
\bigskip
\end{figure}

In words, the proposition states that sufficiently resilient regimes are encouragement environments, and weak regimes are discouragement environments. This result implies that the effect of changing information similarity is qualitatively different based on the resilience of regimes: Increased similarity of information facilitates greater participation against regimes that are hard to change  but can hurt participation against weaker regimes. For instance, consider mass protests. The popular press often claims that mass protests have becomes larger in modern times and attributes this to the change in modern communication technology.\footnote{See, for instance, \cite{economist2020}.}
Our result suggests that this casual observation may be biased: While popular media has highlighted how a more connected world enabled large protests against regimes previously thought to be impregnable, they may have overlooked how this same increased information similarity may have hindered collective action in movements with ex ante easier goals.

\section{Application~1: Success of Mass Protests}\label{sec:resilience of regimes}

We can apply our framework directly to study mass protests. Theorems~\ref{Theorem: CAD helps in encouragement} and~\ref{Theorem: CAD hurts in discouragement} characterize the effect of changes in information similarity on aggregate participation in a mass protest but are silent about the effect on the 
likelihood of a successful protest. 
A priori, increased participation might not imply an increased probability of a successful regime change. 
To see this quickly, suppose that $\overline \strategy$ is the equilibrium in which the probability of successful protest when $\stater=1$ is maximized. For simplicity, restrict attention to symmetric strategies across groups, and let $\overline P$ and $\overline{NP}$ denote respectively the participation and nonparticipation sets. If information becomes more similar, the event $\{\sigr_1 \in \overline P, \sigr_2 \in \overline P\}$ occurs more frequently when $\stater =1$. This increases the probability of a successful protest. However, the probability of the event $\{\sigr_1\in \overline{NP},\sigr_2\in \overline{NP}\}$ also increases, which means it is also more likely that neither group participates. In other words, increased similarity of information increases the probability of the event that both groups participate and the event that neither group participates. The overall effect on the probability of a successful protest is ambiguous. 

To study the effect of increased information similarity on the probability of success of a protest, we specialize to an environment in which the population sizes $\popr_1,\popr_2$ are Poisson distributed with mean $\pop$ and the threshold participation required for regime change is a constant $\overline n$.
We also consider symmetric strategies across the groups.\footnote{In encouragement environments, it is easy to argue that the equilibrium with maximal probability of success is in symmetric strategies. Therefore, we can obtain a similar result as Proposition~\ref{Proposition: probability of success} without restricting to symmetric strategies.} Results analogous to Theorems~\ref{Theorem: CAD helps in encouragement} and \ref{Theorem: CAD hurts in discouragement} about the expected maximal participation still  hold.\footnote{We further show (in Proposition~\ref{Proposition: expected protsize conditional on there being protests} 
in Appendix~\ref{Subsection: appendix size of protests conditional on protests}) that if a strategy profile constitutes a maximal equilibrium under two information structures ranked according to CAD, then \emph{conditional on there being any participation}, the maximal equilibrium participation and the probability of a regime change are strictly higher under the more similar information structure. However, the \emph{unconditional} size of mass protests or probability of successful protests might not  increase with more similar information.} Furthermore, Proposition~\ref{Proposition: probability of success} below shows that increasing similarity can never lower the probability of a successful protest in encouragement environments but can do so in discouragement environments. 

For a given information structure $\jcdf$ and equilibrium strategy profile~$\sigma$ we define
\begin{align*}
    \probsucc(\strategy;\jcdf) :=& \jcdf(\Actionr \ge \threshold+1\vert \stater =1)\\
    \overline \probsucc(\jcdf):=& \max_{\strategy \in \eq(\jcdf)}\probsucc(\strategy).
\end{align*}

\bdefn[Condition $M^\prime$ ]\label{Definition: maximality due to free-riding (probability)} Fix an information structure $\jcdf$. Let $\overline \strategy$ be an equilibrium maximum probability of success ($\bar \Pi(\jcdf) = \Pi(\overline\strategy;\jcdf)$). We say that $\jcdf$ satisfies condition $M^\prime$ if, for any 
strategy profile~$\widehat \strategy$ with $\Pi(\widehat\strategy;\jcdf) > \bar\Pi(\jcdf)$, $\exists \signal \in \widehat P:= P(\widehat\strategy)$, 
such that 
\begin{align*}
 \jcdf^1_\signal(\widehat P) \pivtwo + (1-\jcdf^1_\signal(\widehat P)) \pivone< \frac{c}{\posterior(\signal)}.
\end{align*}
\edefn

\bprop 
\label{Proposition: probability of success} The maximum probability of successful protest increases in encouragement environments, while it decreases in discouragement environments provided it satisfies Condition~$M^\prime$. 
\begin{enumerate}
\item If $\overline n>n^*$ and $\jcdf^1 \cad \jcdfhat^1$, then $\overline \probsucc(\jcdf) \ge \overline\probsucc(\jcdfhat)$. 
\item If $\overline n<n^*$ and $\jcdfhat$ satisfies Condition~$M^\prime$, then $\overline\probsucc(\jcdf) \le \overline\probsucc(\jcdfhat)$. 
\end{enumerate}   
\eprop

\noindent An analogous result holds for welfare (see Appendix~\ref{Appendix: probability of success}).  These findings have important implications for press freedom and authoritarian governments.  \cite{edmond2013information} analyzes a setting in which 
citizens protest and an authoritarian government manipulates the information that citizens see. While the author focuses on the coordination aspect, our setup also considers the free-riding incentive. In today's world, even when an authoritarian government can stop the mainstream media from reporting about unfavorable policies, it is unlikely to prevent people from learning about them. Therefore, the main effect of curbing press freedom is that individuals are less likely to believe that others have seen the same information, or, in other words, information similarity across agents is reduced. Our theory suggests that it is the powerful (that is,    resilient) authoritarian governments that benefit more from curbing press freedom. Less powerful regimes may face larger protests and a higher likelihood of being overthrown when attempting to restrict press freedom.

\section{Application~2: Costly Voting in Committees} 
\label{Section: voting}
Next, we apply our framework to a quite different context: voting in committees \cite[\`a la][]{palfrey1985voter}. Consider members of a monetary policy committee who must vote on whether to raise interest rates. Voting in favor of an interest rate increase is costly for committee members because votes are public and the public wants the status quo to be maintained.  Committee members base their decision on their private information about how a rate increase  affects the economy. An interest rate increase is implemented if and only if a threshold number of members vote in favor of it. In this setting, members' private information can exhibit different levels of similarity of information depending on the diversity of their backgrounds---for example, different academic expertise may lead members to focus on different aspects of the available empirical evidence. What effect do changes in diversity of the committee have on voting outcomes? We can apply our tools to characterize when increased information similarity (less diversity among committee members) strengthens or weakens the incentive to vote in favor of a rate increase based on evidence, and how it affects the choice of the optimal voting threshold. 

Suppose the committee comprises $G>2$ members who must vote for or against a rate increase. 
The economy is in one of two possible states: a rate increase is either unnecessary ($\stater=0$) or warranted ($\stater=1$). Each member $i$ privately receives a noisy binary signal $\sigr_i \in \signalset = \{0,1\}$ about the state $\stater$. Conditional on state $\stater$, the signals are drawn from a joint distribution $\jcdf^\state$. These signals can be interpreted as each committee member's understanding of the available data given their background  and area of expertise. Information among committee members is more similar if their backgrounds are more similar. The rate increase is implemented if and only if more than $\threshold+1$ members cast votes in favor of the rate increase. Intuitively, we say each individual member vote carries little (a lot of) weight if $\threshold$ is high (low). If a rate increase is implemented, the members get a payoff of $\theta$. Voting in favor of a rate increase costs $c$, interpreted as the cost of facing public hostility ex post. For simplicity, we assume $\frac{\cost}{\posterior(0)} > 1$ to guarantee that in any equilibrium, a member who believes that $\sigr_i=0$ never votes in favor of a hike. We also restrict attention to symmetric strategies.  

Notice that now there is no uncertainty about the number of participants or the participation threshold. Further, importantly, with $G>2$ agents, we need a notion of information similarity to compare random variables with $G>2$ dimensions. We extend our notion of CAD as follows. Let $I=\sum_{j\in G}\mathbbm 1_{\sigr_j=1}$ denote the number of committee members who receive signal~$\sigr_j=1$, and let $I_{-i}=\sum_{j\in G\setminus \{i\}}\mathbbm 1_{\sigr_j=1}$ denote the number of members other than $i$ who receive signal~$\sigr_j=1$. 
\begin{definition}\label{Definition: CAD for n dimensions, voting application}
We say $\jcdf^1\cad \jcdfhat^1$ if there exists $k^*\in\{0,1, \ldots, G-2\}$ such that 
$$\gamma^1_1 (k) \le \widehat \gamma^1_1 (k) \ \text{for all} \ k \le k^*$$
$$\text{and} \ \gamma^1_1 (k) \ge \widehat \gamma^1_1 (k) \ \text{for all} \ k > k^*.$$
We call $k^*$ the index of sign change between $\jcdf^1$ and $\jcdfhat^1$.
\end{definition}
To understand the above definition, consider a member who receives a signal $1$. If information similarity increases, then conditional on the rate increase being warranted ($\stater=1$), each member assigns a higher probability to more than $k^*$ others having also observed signal~$1$ 
and a lower probability to fewer than $k^*$ others having observed the opposite signal~$0$.\footnote{Notice that we  do not restrict the conditional beliefs after $\sigr_i=0$. This is because, given $c>\mu(0)$, an agent never casts votes for a rate hike after observing $\sigr_i=0$. For more general environments, we would need restrictions on the conditional beliefs after any signal realization. In a companion paper, we present  extensions of the CAD order for more than two 
dimensions and study their implications for equilibrium behavior in a class of binary-action games.}

Consider a strategy profile, denoted by $\strategy^1$, in which members vote in favor of a rate increase whenever they receive $\sigr_i=1$. 
For $\strategy^1$ to be an equilibrium, we need  
\begin{equation}
\label{ICvoting}
   \gamma^1_1(\threshold)\geq \frac{c}{\mu(1)} 
   \tag{IC-voting},
\end{equation}
where $\gamma^1_1$ is the probability distribution function of  $I_{-i}$ conditional on $\stater=1$ and $\sigr_i=1$, and $\mu(1)$ is the posterior belief upon receiving signal~$\sigr_i=1$. Proposition~\ref{Proposition: voting changes in similarity like theorem 1 and 2}  characterizes how increased similarity of information affects the maximal equilibrium (the equilibrium with the largest number of votes in favor of a rate hike conditional on it being warranted).
The only candidate equilibrium with any votes for a rate hike is $\strategy^1$. Therefore, studying maximal equilibria 
reduces to starting out with a $\strategy^1$ that is an equilibrium and checking whether it remains an equilibrium when information similarity increases. 

\bprop\label{Proposition: voting changes in similarity like theorem 1 and 2}Fix $\threshold$, the threshold~number of votes required to implement a rate increase. Suppose  $\jcdf^1 \cad \jcdfhat^1$, and let $k^*$ be the associated index of sign change between $\jcdf^1$ and $\jcdfhat^1$. Suppose $\strategy^1\in \eq(\jcdfhat)$. 
\begin{enumerate}
\item If $k^* < \threshold$, then $\strategy^1 \in \eq(\jcdf)$.
\item If $k^* \ge \threshold$, then it  is possible that $\strategy^1\notin \eq(\jcdf)$.
\end{enumerate}
\eprop
Let us analyze the incentives of a committee member $i$ who believes that a rate increase is warranted based on her signal~$\sigr_i = 1$. We consider two extreme cases. First, suppose that individual member votes have as little weight as possible---that is, a unanimous vote in favor is required to increase rates (that is, $\threshold +1= \ngroups)$. Then member~$i$'s vote is relevant only when all the others also vote for a hike. This requires that all others also have received the same signal. With increased information similarity, $\gamma^1_1(\ngroups-1)$ increases, which makes it more likely that her vote is relevant. This increases her incentive to vote in favor of a rate increase conditional on the signal~$\sigr_i = 1$, regardless of $k^*$. 
Next, suppose that individual member votes have the highest possible weight---that is, just one vote is required to implement the rate increase ($\threshold +1= 1$). Then member~$i$'s vote is relevant only if all the others vote against the interest rate hike. Increased information similarity means $\gamma^1_1(0)$ decreases. The incentive to vote for a hike in this case is diminished, regardless of $k^*$. In general, how increased information similarity affects a member's incentive to vote will depend on $k^*$. Intuitively, conditional on the rate increase being warranted, a more diverse committee increases (reduces) the maximal equilibrium number of votes in favor of a rate increase if each individual vote has sufficiently high (low) weight.

This leads to an interesting normative question about the design of optimal voting rules in committees: how do changes in the diversity of a committee affect the choice of the optimal voting threshold rule?\footnote{In a recent paper, \cite{kattwinkel2023} study the optimal decision mechanism for juries, allowing for general mechanisms but keeping the information structure fixed with independent signals across jurors. Our framework considers a special class of decision mechanisms, namely those in which a minimum threshold number of votes is required for a decision, and asks how the optimal rule changes with a changing information structure.} The interested reader can refer to Appendix~\ref{Section: appendix optimal voting rule} for details. 

\section{Discussion}\label{Section: robustness}
Finally, we discuss some extensions. Formal results are in the online appendix.

\paragraph{Mixed strategies:} So far we have restricted our attention to pure strategies. This is without loss in encouragement environments because equilibria with maximal participation are necessarily pure and symmetric. However, in discouragement environments, maximal equilibria may be in mixed strategies. 
The main technical challenge of working with mixed strategies is that 
the maximal equilibrium can involve mixing on some signals with different probabilities, and it is no longer possible to partition $\signalset$ into participation and nonparticipation sets~$P$ and $NP$. Much of our analysis examined how $\jcdf^1_\signal(P)$ changes for any $\signal$. But when agents mix with different probabilities after different signals, we need to understand how $\jcdf^1_\signal(\signal')$ changes for each $\signal, \signal'\in \signalset$. Developing tools to completely characterize how the set of mixed-strategy equilibria varies with information similarity is left for future work. However, below, we show that even when we allow for mixing, the main qualitative insight that increasing information similarity can lead to lower maximal equilibrium participation in discouragement environments remains valid. 

For the sake of exposition, suppose  $\popr_1,\popr_2$ are Poisson distributed and $\thresholdr$ is a constant, and suppose  the maximal equilibrium is symmetric and exhibits mixing on exactly one signal. That is, $\signalset$ can be partitioned into three sets: $P,NP$, and $\{\signal^*\}$. $P$ and $NP$ have their usual meanings, while on $\signal^*$ agents mix---that is, participate with probability $\b < 1$. Then, the IC for participation at $\signal^*$ is
\begin{align}
    \jcdfhat^1_{\signal^*}(\signal^*) \ppdf(\threshold, 2\b \pop) + \jcdfhat^1_{\signal^*} (P) \ppdf(\threshold,(1+\b)\pop) +  \jcdfhat^1_{\signal^*} (NP) \ppdf(\threshold,\b\pop) = \frac{\cost}{\posterior(\signal^*)}.\tag{IC-P-Mix}\label{Equation:IC-P-Mix}
\end{align}

Recall that $\ppdf(k,x) = \frac{e^{-x}x^k}{k!}$ is the Poisson probability at $k$ with mean $x$. Suppose that information becomes more similar. How this affects the maximal equilibrium depends on both the type of change in the information structure and how the three pivotal probabilities in \eqref{Equation:IC-P-Mix}---$\ppdf(\threshold,2\b\pop),\ppdf(\threshold,(1+\b)\pop)$, and $\ppdf(\threshold,\b \pop)$---are ranked.

\begin{figure}[h!]
    \centering
    \begin{tikzpicture}
        \begin{axis}[
            width=14cm, height=8cm,
            no markers,
            axis lines=middle,
            xlabel={$x$},
            ylabel={$\ppdf(\threshold, x)$},
            xlabel style={at={(axis description cs:1,0)}, anchor=north},
            ylabel style={at={(axis description cs:-0.07,.5)}, anchor=south},
            xtick=\empty,
            ytick=\empty,
            axis line style={-latex},
            legend style={at={(1,1)}, anchor=north east, draw=none},
            every axis label/.append style={font=\large},
            domain=4.5:18,
            samples=100,
            ymin=-0.013,  
            ymax = 0.14
        ]
            \addplot[ultra thick, blue] {exp(-x) * x^10 / factorial(10)};
            
            \addplot[only marks, mark=*, mark options={scale=1.2, fill=black}] coordinates {(11,0.1193780602280255)};
            \node at (axis cs:11,-0.007) {$\b N$}; 
            
            \addplot[dashed, thick, black] coordinates {(11, 0.1193780602280255) (11, 0)};
            
            \addplot[only marks, mark=*, mark options={scale=1.2, fill=black}] coordinates {(13,0.08587015077390671)};
            
            \node at (axis cs:13,-0.007) {$2\b N$}; 
            
            \addplot[dashed, thick, black] coordinates {(13,0.08587015077390671) (13,0)};
            
            \addplot[only marks, mark=*, mark options={scale=1.2, fill=black}] coordinates {(15,0.048610750829605316)};

            \node at (axis cs:15,-0.007) {$(1+\b) N$}; 
            
            \addplot[dashed, thick, black] coordinates {(15,0.048610750829605316) (15,0)};

            \addplot[only marks, mark=*, mark options={scale=1.2, fill=black}] coordinates {(10,0.1251100357211333)};

            \node at (axis cs:10,-0.007) {$\threshold$}; 
            
            \addplot[dashed, thick, black] coordinates {(10,0.1251100357211333) (10,0)};
            
        \end{axis}
    \end{tikzpicture}
    \caption{Pivotal probabilities}
    \label{fig:poisson-probability}
\end{figure}
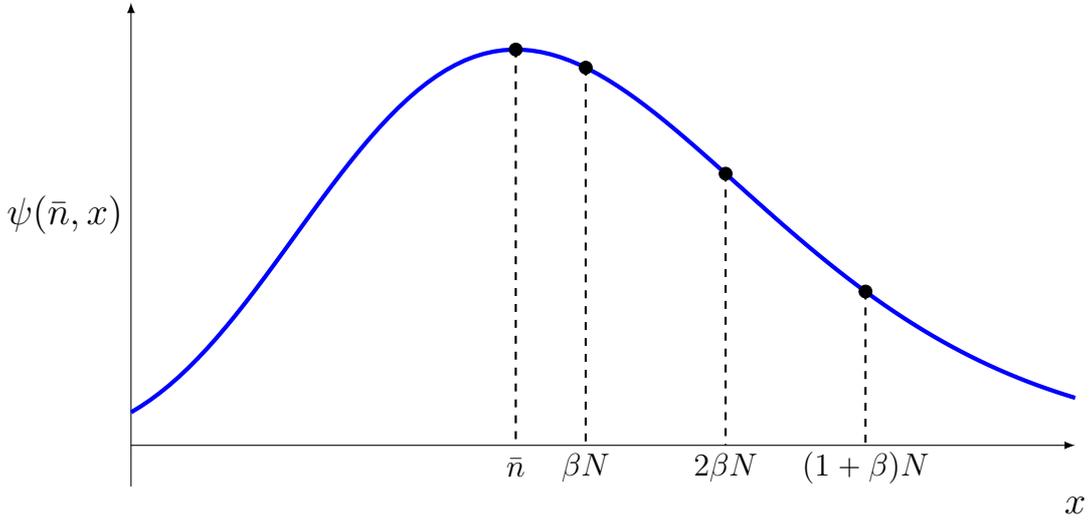

To see this, suppose that $\ppdf(\threshold,\b\pop) > \ppdf(\threshold,2\b\pop)$ (as in Figure~\ref{fig:poisson-probability}) and~$\jcdf^1 \cad \jcdfhat^1$. In particular, suppose that $\jcdf^1_{\signal^*}(P ) = \jcdfhat^1_{\signal^*}(P)$, while $\jcdf^1_{\signal^*}(NP) < \jcdfhat^1_{\signal^*}(NP)$. Then, $\jcdf^1_{\signal^*}$ assigns more weight on $\ppdf(\threshold, 2\b\pop)$ and less on $\ppdf(\threshold,\b\pop)$. Therefore, the LHS of \eqref{Equation:IC-P-Mix} will fall under $\jcdf$, violating the constraint. More importantly, as Figure~\ref{fig:poisson-probability} shows, \eqref{Equation:IC-P-Mix} cannot hold for any higher $\b$ under $\jcdf$.\footnote{It is not necessary that $\b \pop > \threshold$ as in Figure~\ref{fig:poisson-probability} for \eqref{Equation:IC-P-Mix} to not be satisfied for any higher $\b$. If $2\ppdf_2(\threshold,2\b\pop) 
 + \ppdf_1(\threshold,\b\pop) < 0$, then \eqref{Equation:IC-P-Mix} cannot be satisfied for a higher $\b$ under $\jcdfhat$. Here $\ppdf_1(k,x) = \frac{\partial \ppdf(k,x)}{\partial k}$ and $\ppdf_2(k,x) = \frac{\partial \ppdf(k,x)}{\partial x}$.} Therefore, an equilibrium in which players participate after observing signals in $P$ and mix on observing $\signal^*$ must exhibit lower participation under $\jcdf$. Essentially, a player with a signal $\signal^*$ anticipates higher participation under $\jcdf$. In discouragement environments, this weakens his incentive to participate. 

In contrast, suppose that $\jcdf^1_{\signal^*}(NP) = \jcdfhat^1_{\signal^*}(NP)$ and $\jcdf^1_{\signal^*}(P) < \jcdfhat^1_{\signal^*}(P)$. In this case, we assign more weight on $\ppdf(\threshold,2\b\pop)$ and less on $\ppdf(\threshold, (1+\b)\pop)$. Therefore, the LHS of \eqref{Equation:IC-P-Mix} \emph{increases} under $\jcdf$. More importantly, if both $\ppdf(\threshold,2\b\pop)$ and $\ppdf(\threshold,(1+\b)\pop)$ are decreasing in $\b$ (as in Figure~\ref{fig:poisson-probability}), then \eqref{Equation:IC-P-Mix} will hold with an equality for a higher $\b$. Therefore, maximal participation will \emph{increase} provided other equilibrium conditions continue to hold, as will be the case if the change in information similarity is small. Unlike in the previous case, here an agent with a signal $\signal^*$ anticipates lower participation by others under $\jcdf$. In discouragement environments, this increases his incentive to participate. 

The key insight is that the impact of information similarity on incentives is fundamentally linked to players' beliefs about expected participation based on their signals. In discouragement environments, higher participation reduces a player's likelihood of being pivotal, weakening her incentive to participate. When players use mixed strategies, they adjust to the increased expected participation---and the resulting reduced incentive to participate---by lowering their own mixing probability. Thus, whether increased similarity aids or hinders participation in mixed-strategy equilibria depends on how it affects expected participation by others. This insight also applies when agents mix over multiple signals.

\paragraph{State-dependent changes in similarity:} There is a public discourse about how undesirable autocratic regimes restrict the flow of information among citizens, making coordination harder. We can ask how participation in a mass protest changes if the similarity of information is state dependent:  information similarity is not allowed to increase in one state, say, when the regime is autocratic, but increases in the other. 
Our baseline model cannot directly answer this question. We interpret $\stater= 1$ ($\stater=0$) as the feasibility (infeasibility) of collective action, and so changes in information similarity in state $0$ do not affect agents' incentives.  
But our methods can be readily applied to study settings with state-dependent changes by simply relabeling the states---say,  $\states = \{-1,1\}$---so that a regime change is welfare improving in state~$\stater = 1$  and  welfare reducing in  $\stater=-1$. Proposition~\ref{Proposition: general state space} in the online appendix shows that we can derive qualitatively different results in this case: If information similarity increases only in good regimes (when regime change is welfare reducing) and is unchanged in  autocratic regimes (when regime change is welfare improving), then expected participation falls in encouragement environments and increases in discouragement environments.

\paragraph{More than two groups:}  With two groups, comparison in the CAD order allowed for a complete characterization of when more similar information increased or decreased participation. Our application of voting illustrates that our results extend qualitatively to some settings with more than two agents. However, in general, 
with more than two groups, no order of similarity can yield a complete characterization. This is because the probability of being pivotal in a regime-change game is often quasi-concave---first increasing in the number of groups with similar information, and then decreasing---and random variables are ranked according to the quasi-concave order if and only if they have the same distribution.\footnote{To be ranked according to the quasi-concave order means that the expectations of any quasi-concave function are ranked. To see why this ranking implies the same distribution, notice that $f(x):= \ind_{x\ge z}$ and $f(x) := \ind_{x\le z}$ are both quasi-concave.} This means that in general, the participation incentives will not be ranked as we vary information similarity. In the online appendix, we propose a notion of similarity of $n$-dimensional random variables and show 
that results similar to Theorems~\ref{Theorem: CAD helps in encouragement} and \ref{Theorem: CAD hurts in discouragement} are valid, even though the characterization is partial. We use an $n$-dimensional order due to \cite{meyer1990interdependence}, which is stronger than the one in Section~\ref{Section: voting}.

\paragraph{Informativeness of turnout:} In many settings, individuals use protests and petitions to convey dispersed private information to policymakers, and in turn, policymakers use the observed participation in protests or petition to \emph{infer} the state of the world and then decide whether to change policy. Our model abstracts from this, as regime change occurs whenever turnout exceeds an exogenous threshold. We consider a version of our model in which a strategic policymaker observes the realized turnout, updates her belief about the state of the world, and then decides whether to change the regime. Recent works by \cite{battaglini2017public} and \cite{ekmekci2019informal} use a similar setup and focus on the informational role of turnout. We can apply our framework to ask: does increased similarity of information affect the informational content of participation and improve information aggregation? We first show that informativeness of equilibria can decrease with more similar  information.\footnote{This is consistent with current public discourse. For instance, in a piece about technology and protests in the \emph{Atlantic}, Zeynep Tufekci writes, ``Protests are signals: `We are unhappy, and we won’t put up with things the way they are.' But for that to work, the `We won’t put up with it' part has to be credible. Nowadays, large protests sometimes lack such credibility, especially because digital technologies have made them so much easier to organize.''(See \url{https://www.theatlantic.com/technology/archive/2020/06/why-protests-work/613420/})} The intuition is that when information becomes more similar, holding the strategies fixed, the policymaker wants to lower the threshold, and this has two opposing effects. On the one hand, a lower threshold encourages more participation because individuals are more likely to make a difference. On the other hand, a lower threshold exacerbates free-riding. We also show that when the threshold belief at which the policymaker changes the regime is not extreme (in an intermediate range), increasing information similarity can enable information aggregation that would have been impossible under conditionally independent signals.

\paragraph{Information design:} A natural question is how a designer might choose the optimal level of similarity of information, given a certain objective. In the online appendix, we derive the information structure that maximizes expected participation when regime change is beneficial (in $\stater = 1$). We show that in encouragement environments, the optimal information structure is full correlation: both groups receive identical signals. In discouragement environments, interior levels of similarity---neither conditionally independent signals nor full correlation---can be optimal if the conditionally independent signals do not satisfy condition M. Our analysis restricts attention to information structures that are (weakly) more similar than conditionally independent signals. More generally, in discouragement environments, some negative interdependence may be desirable. 

\paragraph{Changing marginal distributions:} 
Throughout we restrict attention to changes in the joint distribution while keeping the marginal distributions unchanged. This assumption ensures that our results are not driven by the changes in the fundamental uncertainty about $\stater$. 
However, in practice, increases in information similarity can simultaneously endow agents with more information about the fundamentals. 
It is straightforward to provide examples when more similar information can lead to lower participation in the discouragement environment (as in Theorem~\ref{Theorem: CAD hurts in discouragement}) even when the marginal distributions change. In Appendix~\ref{Appendix: section on changing information similarity}, we define an order of similarity that does not require marginals to be unchanged. Using that, Proposition~\ref{Proposition: changing marginal result} delivers a similar result to Theorem~\ref{Theorem: CAD helps in encouragement} if we further require that the groups' marginal distributions over posteriors are ranked in the ``more spread-out order'' when $\stater=1$.\footnote{Without this condition, in encouragement environments, information that is both more similar and more Blackwell-informative  informative about $\stater$ can lead to a lower expected participation. Examples are available with the authors.} 

\newpage
\bibliography{ref-coordination.bib}
\end{spacing}
\appendix
\newpage

\section{Appendix: Proofs}
\subsection{Proof of Lemma~\ref{Lemma: Expected protest size}}
\bprf
Let $\sigvec=(x_1,x_2,\ldots x_G)$ be a profile of signal realizations. By definition,  
\begin{align*}
    \protsize(\strategy) =& \sum_{\sigvec \in \signalset^{\ngroups}} \jcdf^1(\sigvec) \left[\sum_{g=1}^{\ngroups} \pop \ind_{\sig_g \in P_g}\right]\\
    =& \sum_{\sigvec \in \signalset^{\ngroups}} \sum_{g=1}^{\ngroups} \jcdf^1(\sigvec)  \pop \ind_{\sig_g \in P_g}\\
    =&\sum_{g=1}^{\ngroups} \sum_{\sig_g \in \signalset} \margdist^1_g(\sig_g)  \pop \ind_{\sig_g \in P_g}\\
    =& \pop\sum_{g=1}^\ngroups \mcdf^1(P_g).
\end{align*}
where the last equality is due to the exchangeability of the distribution.  
\eprf

\subsection{Proof of Proposition~\ref{prop:equilibrium}}

\bprf 
Consider the payoff difference for any agent between participating and not. Let $u_g(a,\signal;\strategy)$ be the expected payoff of an agent from group $g$ by playing action $\action$ given that $\sigr_g =\signal$, and the players are following $\strategy$. We define the net expected payoff from participation as
   $$\De_g(\signal;\strategy):=  u_g(1,\signal;\strategy) - u_g(0,\signal;\strategy)$$
Consider an agent in group $g$. 
Suppose that $\sigr_g = \signal$. If she participates, she incurs a cost $c$ and gets a positive payoff only if the regime change is beneficial ($\stater = 1$) and at least $\thresholdr$ other agents participate. If she does not participate, then she gets a positive payoff only if $\stater = 1$ and the turnout is at least $\thresholdr+1$ without her participation. So, assuming that players play according to $\strategy$, we have  
$$
    \De_g(\signal;\strategy) 
$$
$$    
    = \prob\left(\{\stater = 1\}\bigcap \{\Actionr_{-g} \ge \thresholdr\} \bigg\vert \sigr_g = \signal\right) - c 
    - \prob\left(\{\stater = 1\}\bigcap \{\Actionr_{-g} \ge \thresholdr+1\} \bigg\vert \sigr_g = \signal\right) 
$$
$$    
    = \prob(\stater = 1 \vert \sigr_g = \signal) \Big[ \prob(\Actionr_{-g} \ge \thresholdr \vert \sigr_g = \signal, \stater = 1)
    - \prob(\Actionr_{-g} \ge \thresholdr +1\vert \sigr_g = \signal, \stater = 1)\Big]-c 
$$
$$    
    = \mu(\signal) \prob(\Actionr_{-g} = \thresholdr \vert \stater = 1, \sigr_g = \signal) - c 
$$    
To simplify the above expression further, consider two cases. (i) Suppose $\signal\in P_{g}$. If the realized signal for the other group $\signal_{-g} \in P_{-g}$, then $\Actionr_{-g} = \popr_g-1 + \popr_{-g}$,  and if $\signal_{-g} \notin P_{-g}$, then $\Actionr_{-g} = \popr_g-1$. (ii) Suppose $\signal\notin P_{g}$. Then, $\signal_{-g} \in P_{-g}$ $\implies$ $\Actionr_{-g} = \popr_{-g}$, and $\signal_{-g} \notin P_{-g}$ $\implies$ $\Actionr_{-g} = 0$. Therefore,

\begin{align*}
    \De_g(\signal;\strategy) =& \begin{cases}  
    \posterior(\signal)\Big[\jcdf^1_\signal(P_{-g}) Prob\left(\popr_g-1 + \popr_{-g}= \thresholdr \mid \stater = 1, \sigr_g = \signal\right) \\
    \quad + \jcdf^1_\signal(NP_{-g}) Prob(\popr_g -1= \thresholdr \mid \stater =1, \sigr_g = \signal)\Big]    & \text{ if } \signal \in P_{g} \\
     \posterior(\signal) \jcdf^1_\signal(P_{-g}) Prob( \popr_{-g} = \thresholdr \mid \stater = 1, \sigr_g = \signal)  & \text{ if } \signal \notin P_{g}   \end{cases} 
\end{align*}
Finally, $\strategy$ is an equilibrium if, for all $\signal \in \signalset$ and for all $g\in \{1,2\}$,
\begin{enumerate}
    \item $\strategy_g(\signal) = 1 \implies \De_g(\signal;\strategy) \ge 0$. 
    \item $\strategy_g(\signal) = 0 \implies \De_g(\signal;\strategy) \le 0$. 
\end{enumerate}
The expression in the proposition follows from noting that according to an agent in group $g$, $\popr_g-1 + \popr_{-g}\sim \totalnpmfag(\cdot)$, $\popr_g-1\sim \ngpmfag(\cdot)$, $\popr_{-g}\sim \ngpmf(\cdot)$, and $\thresholdr\sim \tpmf(\cdot)$. 

\eprf 

\subsection{Proof of Theorem~\ref{Theorem: CAD helps in encouragement}}
\bprf 
We prove this result using two steps. In Step 1, we establish that the maximal equilibrium is symmetric. In step 2, we show that the symmetric equilibrium under $\widehat \jcdf^1$ remains an equilibrium under $\jcdf^1$.

\subparagraph{Step 1:} With some abuse of notation, we say that $(P_1, P_2) \in \eq(\jcdf)$ to mean that $\sigma := (\ind_{\sigr_1 \in P_1}, \ind_{\sigr_2 \in P_2}) \in \eq(\jcdf)$.

\blemma\label{Lemma: restrict to symmetric strategies w/o loss in enco} There is a unique maximal equilibrium in encouragement environments, and it is symmetric. \elemma

\bprf

Suppose that $\strategy$ is some asymmetric equilibrium with participation sets $P_1$ and $P_2$ for groups $1$ and $2$ respectively. We show that $\exists$ a symmetric equilibrium $\widehat \strategy$ with a participation set $P \supseteq P_1 \cup P_2$.

For any set $S$, define
\begin{align}
    \mct(S) := S \bigcup \left\{ \signal \in \signalset: \jcdf^1_\signal(S) {\pivother} \ge \frac{c}{\mu(\signal)}\right\}.
\end{align}
In words, $\mct(S)$ adds those signals to $S$ (if there are any) at which an agent wants to participate if he believes that his group will not participate but the other group will participate if they receive a signal in $S$.

\bclaim\label{Claim: operator T satisfies ICP} Let $\strategy = (\ind_{P_1},\ind_{P_2})$ be a strategy profile such that \eqref{Equation: IC P} is satisfied for all $\signal \in P_1 \cup P_2$ given that players follow $\strategy$. Then, for all $\signal \in \mct(P_1 \cup P_2)$, 
\begin{align*}
    \jcdf^1_\signal(\mct(P_1 \cup P_2)) \pivtwo + (1-\jcdf^1_\signal(\mct(P_1 \cup P_2))) \pivone \ge \frac{c}{\posterior(\signal)}
\end{align*}
\eclaim 

\bprf Since $(P_1,P_2) \in \eq(\jcdf)$, \eqref{Equation: IC P} implies, 
\begin{align*}
& \jcdf^1_\signal(P_2) \pivtwo + (1-\jcdf^1_\signal(P_2))\pivone \ge \frac{c}{\posterior(\signal)} & \text{ if } \signal \in P_1\\
    & \jcdf^1_\signal(P_1) \pivtwo + (1-\jcdf^1_\signal(P_2))\pivone \ge \frac{c}{\posterior(\signal)} & \text{ if } \signal \in P_2
\end{align*}
Since $\enco$ and $\jcdf^1_\signal(\cdot)$ is monotonic (in the set inclusion order), 
\begin{align*}
    \jcdf^1_\signal(P_1 \cup P_2) \pivtwo + (1-\jcdf^1_\signal(P_1 \cup P_2)) \pivone \ge \frac{c}{\mu(\signal)} \quad \forall \signal \in P_1 \cup P_2
\end{align*}
If $\mct(P_1 \cup P_2) \neq P_1 \cup P_2$, then, for all $\signal \in \mct(P_1 \cup P_2)\backslash (P_1 \cup P_2)$, we have,  
\begin{align*}
    \jcdf^1_\signal(P_1 \cup P_2) {\pivother} &\ge \frac{c}{\posterior(\signal)}\\
    \implies \jcdf^1_\signal(P_1 \cup P_2) \pivtwo + (1-\jcdf^1_\signal(P_1 \cup P_2) )\pivone &\ge \frac{c}{\posterior(\signal)}
\end{align*}
where the inequality is due to $\enco > 0$ and $\max\{\pivtwo,\pivone\} \ge \pivother$ (see assumption \ref{Assumption: pivotal with one or two groups more likely than pivotal other}). 
\eprf

Define, $\mct^*(P_1 \cup P_2) : = \mct^{\vert \signalset\vert }(P_1,P_2)$. First, notice that $\mct(\cdot)$ is an increasing (in the set-inclusion order) map. Therefore, either $\mct^*(S) = \signalset$ (due to the finiteness of $\signalset$), or $S \subseteq \mct^*(S) \subset \signalset$ for any $S$.  If $\mct^*(P_1 \cup P_2) \neq \signalset$, then, by definition, \eqref{Equation: IC NP} is satisfied for all $\signal \notin \mct^*(P_1 \cup P_2)$ when players play $(\ind_{\sigr_1 \in \mct^*(P_1\cup P_2)},\ind_{\sigr_2\in \mct^*(P_1 \cup P_2)})$. Moreover, \eqref{Equation: IC P} is satisfied when both the groups play $a=1$ on $\mct^*(P_1\cup P_2)$ by Claim \ref{Claim: operator T satisfies ICP}. 
Therefore, given any equilibrium $(P_1, P_2)$, $\mct^*(P_1 \cup P_2)$ is a larger symmetric equilibrium. 

Finally, let $(P,P)$ and $(P',P')$ be two maximal symmetric equilibria with $P \neq P'$. First, Claim \ref{Claim: operator T satisfies ICP} establishes that \eqref{Equation: IC P} is satisfied for all signals in $\mct(P \cup P')$, and hence for all the signals in $\mct^*(P \cup P')$. By construction, \eqref{Equation: IC NP} is satisfied for all the signals outside of $\mct^*(P \cup P')$. Therefore, $\mct^*(P\cup P')$ is an equilibrium, and $P, P' \subseteq \mct^*(P \cup P')$, contradicting the maximality of $P, P'$. Therefore, $P = P'$, i.e., there is a unique maximal equilibrium. 
\eprf

\subparagraph{Step 2:} By definition, $\jcdf^1 \cad \jcdfhat^1$ implies the following:
\begin{enumerate}
    \item $\jcdf^1_\signal(P(\strategy)) \ge \widehat \jcdf^1_\signal(P(\strategy))$ for all $\signal \in P(\strategy)$. 
    \item $\jcdf^1_\signal(P(\strategy)) \le \widehat \jcdf^1_\signal(P(\strategy))$ for all $\signal \in NP(\strategy)$.
\end{enumerate}
Consider ~\eqref{Equation: IC P}, the incentive constraint for an agent who participates. The left hand side is a convex combination of $\pivtwo$ and $\pivone$. In encouragement environments, a higher weight on $\pivtwo$ increases the LHS of \eqref{Equation: IC P}, making the constraint easier to satisfy. Therefore,  
\begin{align*}
     \jcdf^1_\signal(P(\strategy)) \pivtwo + (1-\jcdf^1_\signal(P(\strategy))) \pivone    & \ge   \widehat \jcdf^1_\signal(P(\strategy))    \pivtwo + (1-\widehat \jcdf^1_\signal(P(\strategy))) \pivone\\
    & \ge \frac{c}{\mu(\signal)} \mbox{ } \text{ since }\strategy \in \eq(\widehat\jcdf).
\end{align*}
Therefore, \eqref{Equation: IC P} is satisfied for all $\signal \in P(\strategy)$ under $\jcdf$. 
Similarly, for all $x \in NP(\strategy)$, $\jcdf^1_\signal(P(\strategy)) \le \widehat \jcdf^1_\signal(P(\strategy))$. Therefore, \eqref{Equation: IC NP} is satisfied for all $\signal \in NP(\strategy)$ under signal $\jcdf$. Therefore, $\strategy \in \eq(\jcdf)$.  

Finally, since $\protsize(\cdot)$ depends only the marginal distributions (by Lemma \ref{Lemma: Expected protest size})---which are the same in $\jcdf$ and $\widehat \jcdf$---it follows that $\protsize^*(\jcdf) \ge \protsize^*(\jcdfhat)$. 
\eprf 

\subsection{Proof of Theorem~\ref{Theorem: CAD hurts in discouragement}}

\bprf[Proof of Theorem \ref{Theorem: CAD hurts in discouragement}]
We prove this result using two steps. In step 1, we show that the maximal equilibrium under $\widehat \jcdf^1$ may no longer be an equilibrium under $\jcdf^1$. In step 2, we show that no larger equilibrium can emerge when information becomes more similar. 

\subparagraph{Step 1:} Let $\strategy^*$ be a maximal equilibrium with the associated participation sets $(P_1^*,P_2^*)$.

\begin{enumerate}[\text{Case }1.]
    \item $P_1^* = P_2^*$

Let $P:= P_1^* = P_2^*$. By definition, $\jcdf^1 \cad \jcdfhat^1$ implies the following:
\begin{enumerate}
    \item $\jcdf^1_\signal(P) \ge \jcdfhat^1_\signal(P)$ for all $\signal \in P$.
    \item $\jcdf^1_\signal(P) \le \jcdfhat^1_\signal(P)$ for all $\signal \in NP$.
\end{enumerate}
Consider ~\eqref{Equation: IC P}, the incentive constraint for an agent who participates. The left hand side is a convex combination of $\pivtwo$ and $\pivone$. In discouragement environments, a higher weight on $\pivtwo$ decreases the LHS of \eqref{Equation: IC P}. Therefore, for all $\signal \in P$,
\begin{align*}
     \jcdf^1_\signal(P) \pivtwo + (1-\jcdf^1_\signal(P)) \pivone  
    \le   \jcdfhat^1_\signal(P) \pivtwo + (1-\jcdfhat^1_\signal(P)) \pivone
\end{align*}
Therefore, \eqref{Equation: IC P} may fail for some $\signal \in P$ under $\jcdf$, in which case $\strategy$ may no longer be in $\eq(\jcdf)$.

    \item $P_1^* \neq P_2^*$ 

    Then, $\jcdf^1_\signal(P_i^*) \ge \jcdfhat^1_\signal(P_i^*)$ for all $i \in \{1,2\}$ and $\signal \in P_i^*$. Notice that at least one of $P_1^*\setminus P_2^*$ and $P_2^*\setminus P_1^*$ is not $\emptyset$. Let $P_1^* \backslash P_2^*\neq \emptyset$ wlog. Consider some $\signal \in P_1^* \backslash P_2^*$. Agents in group $2$ must find it incentive compatible to not participate when they receive a signal in $P_1^*$. Therefore, we must have, 
    \begin{align*}
        \jcdf^1_\signal(P_1^*) \pivother \le \frac{c}{\posterior(\signal)}
    \end{align*}
for all $\signal \in P_1^*$. However, since $\jcdf^1_\signal(P_1^*) \ge \jcdfhat^1_\signal(P_1^*)$ for all $\signal \in P_1^*$, \eqref{Equation: IC NP} is harder to satisfy for any $\signal \in P_1^* \backslash P_2^*$, and for any $\signal \in P_2^* \backslash P_1^*$. 
Hence, $\strategy^*$ may no longer be in $\eq(\jcdf)$. 
\end{enumerate}

\subparagraph{Step 2:} Next, consider any $(\widehat P_1, \widehat P_2) \ge (P_1^*,P_2^*)$.  For $(\widehat P_1,\widehat P_2)$ to be an equilibrium, one necessary condition is \eqref{Equation: IC NP} for signals in $\widehat P_1$ and $\widehat P_2$ for groups $2$ and $1$ respectively. That is, we need that, at least one of the following holds: 
\begin{align*}
    \jcdf^1_\signal(\widehat P_1) \pivother \le \frac{c}{\posterior(x)} & \text{ if } \signal \in \widehat P_1 \\
    \jcdf^1_\signal(\widehat P_2) \pivother \le \frac{c}{\posterior(x)} & \text{ if } \signal \in \widehat P_2 
\end{align*}

Suppose that $\widehat P_1 \cap \widehat P_2 = \emptyset$. Since $\jcdfhat$ satisfies condition M (Definition \ref{Definition: maximality due to free-riding}), part $(M2)$ implies that, for some $\signal \in \widehat P_1 \cup \widehat P_2$, 
\begin{align*}
    \jcdfhat^1_\signal(\widehat P_1) \pivother > \frac{c}{\posterior(x)} & \text{ if } \signal \in \widehat P_1 \\
    \jcdfhat^1_\signal(\widehat P_2) \pivother > \frac{c}{\posterior(x)} & \text{ if } \signal \in \widehat P_2 
\end{align*}

By CAD, $\jcdf^1_\signal(\widehat P_1) \ge \jcdfhat^1_\signal(\widehat P_1)$ if $\signal \in \widehat P_1$ (and analogously for $\widehat P_2$). Therefore, $(\widehat P_1,\widehat P_2)$ cannot be an equilibrium in $\jcdf$ if $\widehat P_1 \cap \widehat P_2 = \emptyset$. 

Suppose that $\widehat P_1 \cap \widehat P_2 \neq \emptyset$. Since $\jcdfhat$ satisfies Condition M, if $(M1)$ holds for some $\signal \in \widehat P_1 \cap \widehat P_2$, then, 
\begin{align*}
    \min_{i \in \{1,2\}} \left\{\jcdfhat^1_\signal(\widehat P_i) \pivtwo + (1-\jcdfhat^1_\signal(\widehat P_i)) \pivone\right\} < \frac{c}{\posterior(\signal)}
\end{align*}
By CAD, $\jcdf^1_\signal(\widehat P_i) \ge \jcdfhat^1_\signal(\widehat P_i)$ for $i\in \{1,2\}$. Since $\disco$, this implies that,\begin{align*}
    \min_{i\in \{1,2\}} \left\{ \jcdf^1_\signal(\widehat P_i) \pivtwo + (1-\jcdf^1_\signal(\widehat P_i)) \pivone\right\} < \frac{c}{\posterior(\signal)}
\end{align*}
Therefore, \eqref{Equation: IC P} fails for such an $\signal$.

Finally, if \eqref{Equation: IC P} is satisfied for all $\signal \in \widehat P_1 \cap \widehat P_2$, then, by Condition M $(ii)$, the exact same argument as in the case when $\widehat P_1 \cap \widehat P_2 = \emptyset$ implies that $(\widehat P_1, \widehat P_2)$ cannot be an equilibrium under $\jcdf$. Therefore, no larger equilibrium can exist under $\jcdfhat$, i.e., $\protsize^*(\jcdfhat) \le \protsize(\jcdf)$. 
\eprf

\subsection{Proof of Proposition \ref{Proposition: probability of success}}
\label{Appendix: probability of success}

We first derive a useful property of the CAD order. \blemma\label{Lemma: CAD means transpose of size alpha} If $\ydist \cad \yhatdist$, then, for every $T \subseteq \mathcal Y$, $\exists$  $\a_T \ge 0$ such that $\ydist(T,T) =  \yhatdist(T,T) + \a_T$, and $\ydist(T,\mathcal Y\backslash T) = \yhatdist(T,\mathcal Y\backslash T) - \a_T$.\elemma

\noindent

\bprf[Proof of Lemma~\ref{Lemma: CAD means transpose of size alpha}]
Suppose $\mathcal Y\subset \mathbb R$ is finite, and $Y$ and $\widehat Y$ are two~$\mathcal Y^2$-valued random variables with joint distributions $\ydist$ and $\yhatdist$ respectively, and identical marginals. Consider any two distinct points in the support of $Y$, say $y_j, y_k$.  Define an ``elementary transformation along identical intervals'' (ETI) as an operation in which, for some $\a > 0$, we increase the probability mass on points $(y_j,y_j)$ and $(y_k,y_k)$ each by $\a$, and reduce the probability mass on $(y_j,y_k)$ and $(y_k,y_j)$ each by $\a$. An alternative characterization of our CAD order in two dimensions is that $\ydist \cad \yhatdist$ if and only if $\ydist$ can be derived from $\yhatdist$ by a finite sequence of ETIs. This follows from \cite{meyer1990interdependence} (Proposition 1).  We use this characterization to establish Lemma~~\ref{Lemma: CAD means transpose of size alpha}. 

Let $(y_{i,k},y_{j,k})_k$, $k=1,\ldots n$, be a finite set of points in $\mathcal Y^2$ describing a sequence of ETIs, each with a mass $\a_k$, to obtain $\ydist$ from $\yhatdist$. Let the resulting distribution after the $k$-th ETI be $\yhatdist_k$. So, 
$\yhatdist_1 = \yhatdist$ and $\yhatdist_n = \ydist$. If $(y_{i,k},y_{j,k}) \in T\times T$ or $(\mathcal Y \backslash T) \times (\mathcal Y \backslash T)$, then $\yhatdist_k(T,T) = \yhatdist_{k-1}(T,T)$. On the other hand, if exactly one of $\{y_{i,k},y_{j,k}\}$ is in $T$ for some $k$, then $\yhatdist_k(T,T) = \yhatdist_{k-1}(T,T) + \a_k$. Therefore, $\ydist(T,T) = \yhatdist(T,T) + \sum_{k=1}^n \a_k$. Since an ETI leaves the marginal distribution unchanged, therefore $\ydist(T, \mathcal Y\backslash T) = \yhatdist(T,\mathcal Y\backslash T) - \sum_{k=1}^n \a_k$. The lemma follows. 

\eprf 

\bprf[Proof of Proposition~\ref{Proposition: probability of success}]
Towards establishing $(1)$, suppose $\enco$ and $\jcdf^1\cad \jcdfhat^1$. First, by the argument in Theorem~\ref{Theorem: CAD helps in encouragement}, $\eq(\jcdfhat) \subseteq \eq(\jcdf)$. Consider any equilibrium $\strategy \in \eq(\jcdfhat)$, and let the associated participation and nonparticipation sets be P and NP respectively. By Lemma \ref{Lemma: CAD means transpose of size alpha}, $\jcdf^1(P,P) = \jcdfhat^1(P,P) + \a$ and $\jcdf^1(P,NP) = \jcdfhat^1(P,NP) - \a$ for some $\a\ge 0$.\footnote{$\jcdf^1(T,T)$ is a shorthand to denote $\jcdf^1(\{\sigr_1 \in T, \sigr_2\in T\})$.}  Therefore, 
\begin{align*}
    \probsucc(\strategy;\jcdf) =& (1-\pcdf(\threshold,2\pop))\jcdf^1(P,P) + 2((1-\pcdf(\threshold,\pop))\jcdf^1(P,NP) \\
    =&(1-\pcdf(\threshold,2\pop))(\jcdfhat^1(P,P) +\a)    +  2((1-\pcdf(\threshold,\pop))(\jcdfhat^1(P,NP)-\a) 
\end{align*}
Simplifying, we get
\begin{align*}
   \probsucc(\strategy;\jcdf) - \probsucc(\strategy;\jcdfhat) 
    =& \a \left( 2 \pcdf(\threshold,\pop) - \pcdf(\threshold,2\pop) - 1 \right)
\end{align*}
It is easy to check that $\frac{\partial \pcdf(k,x)}{\partial x} = - \ppdf(k,x)$. Defining $\De(N):= 2 \pcdf(\threshold,\pop) - \pcdf(\threshold,2\pop) -1$, we have $\De'(N) = 2 ( \ppdf(\threshold,2\pop) - \ppdf(\threshold,\pop) > 0$ by assumption. Also, $\lim_{\pop\to 0} \De(\pop) = 0$. Therefore, $\De(\pop) > 0$ for any $\pop$ such that $\enco$. Hence, $\probsucc(\strategy;\jcdf) \ge \probsucc(\strategy;\jcdfhat)$ with the inequality being strict whenever $\a > 0$. Therefore, $\overline\probsucc(\jcdf)\ge \overline\probsucc(\jcdfhat)$.  

Towards establishing $(2)$, suppose that $\overline\probsucc(\jcdf) > \overline\probsucc(\jcdfhat)$, $\jcdf^1 \cad \jcdfhat^1$, and $\jcdfhat^1$ satisfies Condition M (Probability). Let $\strategy^*$ and $\widehat\strategy^*$ be the equilibria with maximum probabilities of success under $\jcdf$ and $\jcdfhat$ respectively. 

First, we argue that $\probsucc(\strategy^*;\jcdfhat) > \probsucc(\widehat\strategy^*;\jcdfhat)$. Suppose not. By the previous argument, we have that, \emph{for any strategy profile $\strategy$,}, $\probsucc(\strategy; \jcdf) < \probsucc(\strategy;\jcdfhat)$ whenever $\disco$. Therefore, 
\begin{align*}
    \probsucc(\widehat\strategy^*;\jcdfhat) < \probsucc(\strategy^*;\jcdf)\le \probsucc(\strategy^*;\jcdfhat) \le \probsucc(\widehat\strategy^*;\jcdfhat);
\end{align*}
a contradiction. 

But, if $\probsucc(\strategy^*;\jcdfhat) > \probsucc(\widehat\strategy^*;\jcdfhat)$, then, by Condition-M (Probability), $\strategy^* \notin\eq(\jcdf)$. Therefore, it cannot be an equilibrium with maximum probability under $\jcdf$. 
\eprf

While Proposition~\ref{Proposition: probability of success} characterizes the effect of increasing interdependence on the probability of success when $\stater=1$, it raises a closely related question of how increasing similarity affects overall welfare---the ex-ante expected utility of a representative agent in any group~$g$.

To this end, define welfare given an information structure and a strategy profile, and maximal welfare achievable in any equilibrium~$\welfare^*$ as follows.\footnote{Exchangeability of the distribution of $\sigr$ implies that welfare does not depend on group identity.}
The following is a corollary of Proposition~\ref{Proposition: probability of success}
\begin{align*}
    \welfare(\strategy,\jcdf):=& \mu_1 \probsucc(\strategy;\jcdf) - c \E_{\jcdf}\left[\strategy(\sigr_g)\right]\\
    \welfare^*(\jcdf):=&\max_{\strategy \in \eq(\jcdf)} \welfare(\strategy,\jcdf)
\end{align*}
The above expression makes it clear that welfare unambiguously increases whenever $\enco$ since $\probsucc(\strategy;\cdot)$ increases with similarity for any $\strategy$ due to Proposition~\ref{Proposition: probability of success}. Analogous reasoning (along with an analogous Condition M (Probability) replaced by Condition M (Welfare)) establishes that welfare, too, decreases under discouragement environments. We choose not to present it as a formal proposition. 

\subsubsection{Turnout conditional on protests}\label{Subsection: appendix size of protests conditional on protests}
In the main text, we focus on how information similarity affects the equilibrium participation. However, even when we fix an equilibrium, information similarity makes the participation more coordinated. Accordingly, conditional on there being a protest, we may see that the protests are more likely to bring about social changes. The following proposition formalizes this intuition.

\bprop\label{Proposition: expected protsize conditional on there being protests}
Suppose $\jcdf^1 \cad \jcdfhat^1$, and $\strategy \in \eq(\jcdf)\cap \eq(\jcdfhat)$ and $\strategy$ is symmetric. Then, $$\jcdf^1\left[\Actionr(\strategy)> \threshold\Big\vert  \Actionr(\strategy) > 0\right] \ge \jcdfhat^1\left[\Actionr(\strategy)>\threshold\Big\vert \Actionr(\strategy) > 0\right].$$
\eprop 
\bprf 
Let $H_2(\cdot), H_1(\cdot)$ be the CDFs of $\popr_1 + \popr_2$ and $\popr_1$ (and $\popr_2$ by symmetry) respectively. Suppose that $\strategy \in \eq(\jcdf) \cap \eq(\jcdfhat)$ and $\jcdf^1 \cad \jcdfhat^1$. Under $\jcdfhat$, the probability of there being a protest at all when $\stater=1$ is $\jcdf^1(\{\Actionr > 0\}) = 1 - \jcdfhat^1(NP,  NP)$. Here $\jcdf^1(NP,NP)$ means $\jcdf^1(\{\sigr_1 \in NP(\strategy), \sigr_2\in NP(\strategy)\})$. 
Therefore, 
\[\jcdf^1\left(\{\Actionr(\strategy)> \threshold)\} \Big\vert \{\Actionr(\strategy) > 0\}\right) =   \frac{\left[(1-H_2(\threshold)) \jcdf^1(P,P) + 2 (1-H_1(\threshold)) \jcdf^1(P,NP)\right]}{1-\jcdf^1(NP,NP)}.\]
Then, by Lemma~\ref{Lemma: CAD means transpose of size alpha}, there exists $\alpha>0$ such that 
\begin{align}
  &\jcdf^1\left(\{\Actionr(\strategy)> \threshold)\} \Big\vert \{\Actionr(\strategy) > 0\}\right) &   \nonumber \\
    & = \frac{\left[(1-H_2(\threshold)) (\jcdfhat^1(P,P)+\a) + 2 (1-H_1(\threshold)) (\jcdfhat^1(P,NP)-\a)\right]}{1-\jcdfhat^1(NP,NP)-\a}& \nonumber\\
     & = \underbrace{\frac{\jcdfhat^1(P,P) + \a}{1-\jcdfhat^1(NP,NP) - \a}}_{\textrm{Increasing in } \alpha} \left( 1-H_2(\threshold)\right) +\underbrace{\left(1-\frac{\jcdfhat^1(P,P) + \a}{1-\jcdfhat^1(NP,NP) - \a}\right)}_{\textrm{Decreasing in }\alpha}\left(1-H_1(\threshold)\right))& \nonumber
     \end{align}
The equality follows from noting that $(\jcdfhat^1(P,P) + \a) + 2 (\jcdfhat^1(P,NP) - \a) = 1-\jcdfhat^1(NP,NP) -\a$. Notice that the above expression is a convex combination of $1-H_2(\threshold)$ and $1-H_1(\threshold)$. Since $\popr_1, \popr_2 \ge 0$ a.s., $H_2(\cdot) \le H_1(\cdot)$. Therefore, the LHS puts a larger weight on the larger term for any $\a \ge 0$, which implies   
     \begin{align}
    &\jcdf^1\left(\{\Actionr(\strategy)> \threshold)\} \Big\vert \{\Actionr(\strategy) > 0\}\right) &   \nonumber \\
      &>\frac{\jcdfhat^1(P,P) }{1-\jcdfhat^1(NP,NP) } \left( 1-H_2(\threshold)\right) +\left(1-\frac{\jcdfhat^1(P,P) }{1-\jcdfhat^1(NP,NP) }\right)\left(1-H_1(\threshold)\right)&\nonumber\\
      &= \jcdfhat^1\left(\{\Actionr(\strategy)> \threshold)\} \Big\vert \{\Actionr(\strategy) > 0\}\right). \nonumber
\end{align}

\eprf

\subsection{Proof of Proposition  \ref{Proposition: voting changes in similarity like theorem 1 and 2}}
\bprf 
Suppose that $k^* < \threshold$. Then, $\strategy^1\in \eq(\jcdfhat)$ implies that $\widehat\g^1_1(\threshold) \ge \frac{c}{\posterior(1)}$. By CAD, $\g^1_1(\threshold) \ge \widehat \g^1_1(\threshold)$. Therefore, \eqref{ICvoting} continues to be satisfied under $\jcdf$, and hence $\strategy^1 \in \eq(\jcdf)$. 

On the other hand, if $k^* \ge \threshold$, then $\g^1_1(\threshold) \le \widehat \g^1_1(\threshold)$. Therefore, it is possible that $\strategy^1 \in \eq(\jcdfhat)$ but $\strategy^1\notin \eq(\jcdf)$. Hence, part $(2)$ of the proposition follows. 
\eprf

\subsection{Optimal voting rule from Section~\ref{Section: voting} }\label{Section: appendix optimal voting rule}

Since there are multiple equilibria, we assume the maximal equilibrium is played, and the optimal vote threshold is one that maximizes the probability that a rate increase is implemented conditional on it being warranted; that is, the optimal threshold $\threshold^*(\jcdf)$ is given by$$\threshold^*(\jcdf) := \argmax_{\threshold} \jcdf^1(I\geq \threshold+1|\stater=1),$$ subject to the \eqref{ICvoting} constraint.\footnote{The qualitative argument is unchanged if we assume a small negative payoff from raising rates when not necessary to do so. Essentially, this formulation assumes that the cost of inflation due to failure to raise rates when required far outweighs a contemporaneous loss in output.} In the absence of the incentive constraint, the lowest possible $\threshold$ would be optimal. However, a low $\threshold$ reduces an individual member's incentive to vote in favor of a rate increase (conditional on private information and the increase being warranted). So the optimal rule is the lowest vote threshold that satisfies the incentive constraint. We can ask how this optimal threshold varies with committee diversity.

\bprf Consider $k^* < \threshold^*(\jcdfhat)$. Then, by definition, 
$$\gamma^1_1(\threshold^*(\jcdfhat))\ge \widehat\gamma^1_1(\threshold^*(\jcdfhat))\ge \frac{\cost}{\posterior(1)}.$$ 
This means under the same policy threshold $\threshold^*(\jcdfhat)$, the incentive constraint is satisfied even under more similar experiences ($\jcdf$). Since the designer's objective $\jcdf^1(I\ge \threshold+1)$ is decreasing in $\threshold$, we have $\threshold^*(\jcdf)\le \threshold^*(\jcdfhat)$.

Next, consider $k^* \ge \threshold^*(\jcdfhat)$.    Recall that $\threshold^*(\jcdfhat)$ is the lowest $\threshold$ that satisfies the incentive constraint under $\jcdfhat$. Therefore, for any $\threshold <\threshold^*(\jcdfhat)$, 
$$ \widehat\gamma^1_1(\threshold)<\frac{\cost}{\posterior(1)}.$$
Since $k^* \ge \threshold^*(\jcdfhat)>\threshold$, by definition,
$$\gamma^1_1(\threshold) \le \widehat\gamma^1_1(\threshold)<\frac{\cost}{\posterior(1)}.$$
This means for any policy $\threshold<\threshold^*(\jcdfhat)$, under more similar experiences  ($\jcdf$), the incentive constraint does not hold. Moreover, since $\gamma^1_1(\threshold^*(\jcdfhat)) \le \widehat\gamma^1_1(\threshold^*(\jcdfhat))$, even under policy $\threshold^*(\jcdfhat)$, the incentive constraint may no longer be satisfied under more similar experiences  ($\jcdf$). Therefore, $\threshold^*(\jcdf) \ge \threshold^*(\jcdfhat)$. 
\eprf 
\newpage
\section{Online Appendix (Not for publication)}

\subsection{On state-dependent changes in similarity}\label{Appendix: state dependent similarity}
Suppose that $\states = \{-1,1\}$. We keep the model unchanged in all respects otherwise. Now a regime change when $\stater = -1$ is welfare-reducing. We demonstrate how increasing similarity in state  $\stater = -1$ has the opposite effects from what we showed in the main paper. To this end, we specialize to symmetric strategies. Then, given any strategy $\strategy$, we have the associated participation and nonparticipation sets given by $P(\strategy)$ and $NP(\strategy)$ respectively.  

For $\strategy$ to be an equilibrium, the IC constraints for protesting and not-protesting are: 
\begin{align}
    &\posterior(\signal) \left[\jcdf^1_\signal(P)\pivtwo + (1-\jcdf^1_\signal(P)) \pivone \right] \nonumber\\
    -&(1-\posterior(\signal))\left[ \jcdf^{-1}_\signal(P)\pivtwo + (1-\jcdf^{-1}_\signal(P)) \pivone\right] \ge c 
& 
\text{ if } \signal \in P \tag{IC:P-S} \label{Equation: IC-P general state space}
\end{align} 
\begin{align}
    &\posterior(\signal) \jcdf^1_\signal(P)\pivother -(1-\posterior(\signal)) \jcdf^{-1}_\signal(P)\pivother \le c & \text{ if } \signal \notin P \tag{IC:NP-S} \label{Equation: IC-NP general state space}
\end{align}
The only difference from our benchmark setup is the second term in the incentive constraints. This captures the probability of being pivotal in state $\stater=-1$.
It is straightforward to see that an increase in similarity in state~$1$ (i.e., CAD increases of $\jcdf^1$) has the same impact as in the main paper (for the natural modification of Condition M for this environment). But now consider the effects of increases in similarity in $\jcdf^{-1}$. We can interpret
$$\jcdf^{-1}_\signal(P)\pivtwo + (1-\jcdf^{-1}_\signal(P)) \pivone$$ as the cost of making a difference 
in state $\state = -1$ for a participant. In encouragement (discouragement) environments, an increase in similarity increases (decreases) this cost, thus reducing (increasing) the incentive for participation among participants. In other words, CAD increase of $\jcdf^{-1}$ has the opposite impact, compared to CAD increases of $\jcdf^1$, on the incentive of the participants (\eqref{Equation: IC-P general state space}). 

For nonparticipants ($x\notin P$), higher similarity in state $\stater=1$ reduces the LHS in \eqref{Equation: IC-NP general state space} while higher similarity in state $\stater=-1$ increases it. Under the following assumption, the incentive constraint of the nonparticipants is always satisfied regardless of $\jcdf^{-1}$. 
\begin{assumption}
\label{assumption:noNPviolation}
For any $\strategy$ with $\protsize(\strategy) > \protsize(\strategy^*)$, and any $\signal \notin P(\strategy)$,  $$\posterior(\signal)\jcdf^1_\signal(P(\strategy))\pivother < c.$$
\end{assumption}
We can again use CAD to characterize the effect of information similarity. 
\bprop\label{Proposition: general state space}
Suppose $\jcdfhat^1$ satisfies assumption \ref{assumption:noNPviolation}. Let $\jcdf:=(\jcdfhat^1,\jcdf^{-1})$ be an information structure such that $\jcdf^{-1} \cad \jcdfhat^{-1}$, then 
\begin{enumerate}
    \item $\protsize^*(\jcdf) \ge \protsize^*(\jcdfhat)$ if $\disco$. And,
    \item $\protsize^*(\jcdf) \le \protsize^*(\jcdfhat)$ if $\enco$.
\end{enumerate}
\eprop 

\bprf Let $\strategy^*$ be a maximal equilibrium under $\jcdf$ with associated participation and nonparticipation sets $P^*$ and $NP^*$ respectively. 

Suppose that $\pivone > \pivtwo$. Then, for all $\signal \in P^*$
\begin{align*}
    \jcdf^{-1}_\signal(P^*) \pivtwo + (1-\jcdf^{-1}_\signal(P^*)) \pivone \le \jcdfhat^{-1}_\signal(P^*) \pivtwo + (1-\jcdfhat^{-1}_\signal(P^*)) \pivone 
\end{align*}
And, for all $\signal \notin P^*$, $\posterior(\signal) \jcdf^1_\signal(P^*) < \cost$ by Assumption \ref{assumption:noNPviolation}. Therefore, $\strategy^* \in \eq(\jcdf)$, proving part $(1)$. 

For $(2)$, for all $\signal \in P^*$, \begin{align*}
    \jcdf^{-1}_\signal(P^*) \pivtwo + (1-\jcdf^{-1}_\signal(P^*)) \pivone \ge \jcdfhat^{-1}_\signal(P^*) \pivtwo + (1-\jcdfhat^{-1}_\signal(P^*)) \pivone. 
\end{align*}
Therefore, $\strategy^*$ may not be in $\eq(\jcdf)$. Finally, consider any strategy profile $\strategy$ such that $\protsize(\strategy) \ge \protsize(\strategy^*)$. By Assumption \ref{assumption:noNPviolation}, \eqref{Equation: IC-NP general state space} holds for all $\signal \in P$ in both, $\jcdfhat$ and $\jcdf$ (Since $\jcdf^1=\widehat \jcdf^1$). Since $\strategy^*$ is a maximal equilibrium under $\jcdfhat$, $\strategy \notin \eq(\jcdfhat)$. Therefore, $\exists \signal \in P$ for whom \eqref{Equation: IC-P general state space} fails. Finally, for all $\signal \in P$, \begin{align*}
    \jcdf^{-1}_\signal(P) \pivtwo + (1-\jcdf^{-1}_\signal(P)) \pivone \ge \jcdfhat^{-1}_\signal(P) \pivtwo + (1-\jcdfhat^{-1}_\signal(P)) \pivone. 
\end{align*}
Therefore, $\strategy \notin \eq(\jcdf)$. The proposition follows.

\eprf

\subsection{On more than 2 groups}\label{Section: more than two groups}
To establish that the qualitative results extend to settings with more than two groups, we consider a specialized environment in which the group sizes $\popr_g$ are drawn from Poisson distributions with means $N$ and the resilience $\thresholdr$ is a deterministic $\threshold$, and restrict ourselves to symmetric equilibrium. Let  $\psi(k,N)=\frac{\exp(-N)N^k}{k!}$ be the probability that nature chooses a group to have $k$ agents. We use the following strong notion of similarity for $n$-dimensional random variables with $n>2$, given by \cite{meyer1990interdependence}. Let $\mc Y$ be a finite subset of $\real^n$, and let 
$\scry$ be a set of $\mc Y$-valued random variables with exchangeable distributions, i.e., the marginal distribution of $Y_i$ and $Y_j$ are equal for any $i,j$. 

\bdefn\label{Definition: Meyer n dimensional similarity}Let $Y,\yhat \in \scry$ be two random variables with distributions $\ydist,\yhatdist$ respectively. Then, $Y \scad \yhat$ if, $$\ydist(\{Y_1 = y_{i_1},Y_2 = y_{i_2},\ldots,Y_n=y_{i_n}\}) \le \yhatdist(\{\yhat_1 = y_{i_1},\yhat_2 = y_{i_2},\ldots,\yhat_n=y_{i_n}\})$$ for all $(i_1,\ldots,i_n)$ for which it is not the case that $i_1 = i_2=\ldots=i_n$.\edefn 
With this order, we can obtain results analogous to Theorem \ref{Theorem: CAD helps in encouragement} and \ref{Theorem: CAD hurts in discouragement}. 

\bprop\label{Proposition: n dimensional result with Meyer order}Suppose that $\jcdf^1\scad \jcdfhat^1$. 
\begin{enumerate}
\item If $\ppdf(\threshold, \ngroups \pop) \ge \ppdf(\threshold,k\pop)$ for all $k < \ngroups -1$, then $\protsize^*(\jcdf) \ge \protsize^*(\jcdfhat)$. 
\item If $\ppdf(\threshold,\ngroups \pop) < \ppdf(\threshold,k\pop)$ for all $k< \ngroups-1$, and $\jcdfhat$ satisfies condition M, then $\protsize^*(\jcdf)\le \protsize^*(\jcdfhat)$. 
\end{enumerate}
\eprop

\bprf 
For any set $T \subset \signalset$, we define $I(T)$ as the random variable denoting the number of other groups that receive a signal $x\in T$. Formally, for any group $g=1$ (say), $I(T):= \sum_{g\neq 1} \ind_{\sigr_g \in T}$. 
Abusing notation, we define $\jcdf^1_x(T,k)= \jcdf(I(T)=k \vert \sigr_1 = \signal,\stater=1)$ as the probability that $k=0,1,\ldots G-1$ other groups see a signal in $T$ when group $g=1$ (say) sees the signal $x$ and the state is $\stater =1$. Notice that $\jcdf^1 \scad \jcdfhat^1$ implies that 
\begin{enumerate}
    \item For $x\in T$, $\jcdf^1_x(T,k)\leq \jcdfhat^1_x(T,k)$ for all $k=0,1,2\ldots G-2$ and  $\jcdf^1_x(T,G-1)\geq \jcdfhat^1_x(T,G-1)$.
    \item For $x\notin T$, $\jcdf^1_x(T,k)\leq \jcdfhat^1_x(T,k)$ for all $k=0,1,2\ldots G-1$
\end{enumerate}

Let $\strategy^*$ be the maximal participation equilibrium in $\jcdfhat$, and let $P^*$ and $NP^*$ be the associated participation and not-participation sets respectively. 
For more than 2 groups, the IC for protesting and not-protesting can be modified as follows
\begin{align}
    \sum_{k=0}^{\ngroups-1}
    \jcdfhat^1_x(T,k)\ppdf(\threshold, (k+1)\pop) &\ge \frac{c}{\posterior(\signal)} \quad \text{ if } \signal \in P^* \tag{IC:P-G}\label{Equation: IC P more than 2 groups}\\
    \sum_{k=0}^{\ngroups-1}
    \jcdfhat^1_x(T,k)\ppdf(\threshold, k\pop) &\le \frac{c}{\posterior(\signal)} \quad \text{ if } \signal \in NP^* \tag{IC:NP-G}\label{Equation: IC NP more than 2 groups}
\end{align}
It is easy to see that when the similarity increases ($\jcdf^1\scad \jcdfhat^1$), if $\ppdf(\threshold, \ngroups \pop) \ge \ppdf(\threshold,k\pop)$ for all $k < \ngroups -1$, then the LHS increases in \eqref{Equation: IC P more than 2 groups} and decreases in \eqref{Equation: IC NP more than 2 groups}. Therefore, $\strategy^* \in \eq(\jcdf)$ and accordingly $\protsize^*(\jcdf)\ge \protsize^*(\jcdfhat)$. On the other hand, if $\ppdf(\threshold, \ngroups \pop) \le \ppdf(\threshold,k\pop)$ for all $k < \ngroups -1$, then under $\jcdf^1$ the LHS decreases in \eqref{Equation: IC P more than 2 groups}, which can violate the incentive of the participant, and accordingly, $\strategy^*$ may no longer be an equilibrium. Moreover, given that $\jcdfhat$ satisfies Condition M, given any $\strategy$ such that $P(\strategy) \ge P^*$, $\exists \signal \in P(\strategy)$ such that, 
\begin{align*}
\sum_{k=0}^{\ngroups-1}\jcdfhat^1_x(T,k) \ppdf(\threshold, (k+1)\pop) &< \frac{c}{\posterior(\signal)}\\
\implies 
\sum_{k=0}^{\ngroups-1}\jcdf^1_x(T,k) \ppdf(\threshold, (k+1)\pop) &< \frac{c}{\posterior(\signal)}.
\end{align*}
Therefore, such $\strategy$ cannot constitute an equilibrium under $\jcdf^1$, and accordingly, $\protsize^*(\jcdf)\le \protsize^*(\jcdfhat)$ with the inequality being strict whenever $\strategy^* \notin \eq(\jcdf)$. 
\eprf

\subsection{Informativeness of turnout}\label{Section: informativeness of turnout}

Consider our baseline model with two groups, two states, and a finite set of signals. For simplicity, we assume that $\popr_1,\popr_2$ are Poisson distributed with mean $\pop$; that is, let  $\psi(k,N)=\frac{\exp(-N)N^k}{k!}$ be the probability that nature chooses a group to have $k$ agents. However, now assume that a strategic policymaker observes the realized turnout and then decides whether to change the regime. We restrict attention to symmetric strategies, and as before, we denote the associated participation and nonparticipation sets of any strategy by $P$ and $NP$, respectively. 

Given an information structure $\jcdf$, agents' strategy~$\strategy$, and aggregate turnout~$\Actionr$, the policymaker's belief about the state of the world is given by the likelihood function $$\beta(\cdot; \jcdf, P) := \frac{Prob(\stater = 1 \vert \Actionr = \cdot, P)}{1-Prob(\stater = 1 \vert \Actionr = \cdot, P)}.$$
The policymaker changes the status quo only if she is sufficiently confident that the state is $1$; that is,  there is a cutoff belief, $\underline \beta$, such that the policymaker changes the status quo if $\beta(k) \ge \underline \beta > 0$.  Therefore, the policymaker's preferences are (ordinally) aligned with the citizens'. 

We define informativeness of turnout, given a strategy $\strategy$, and its associated participation set $P$, as in \cite{ekmekci2019informal}:
\begin{align*}
    I(P):= \mcdf^1(P) - \mcdf^0(P)
\end{align*}
Define $\bar P:= \{ \signal \in \signalset: \mcdf^1(\{\signal\}) > \mcdf^0(\{\signal\})$. It is easy to see that the informativeness of any strategy is bounded $I(\bar P)$. Therefore, for a given information structure $\jcdf$ with fixed marginals, we say that \emph{information aggregates} if $\ind_{\bar P}$ is an equilibrium under $\jcdf$. 

We  fix $\jcdf^0$ to investigate when, if at all, increasing similarity of information facilitates information aggregation. To this end, define 
\begin{align*}
    \underline l &:= \frac{\prior}{1-\prior} \frac{\mcdf^1(\bar P)^2}{\jcdf^0(\bar P,\bar P)}\\
    \bar l & := \frac{\prior}{1-\prior} \frac{\mcdf^1(\bar P)}{\jcdf^0(\bar P,\bar P)}.
\end{align*}

\bprop\label{Proposition: information aggregation with similarity} Suppose that $\mcdf^1(\bar P) \mcdf^0(\bar P) > \jcdf^0(\bar P,\bar P)$.\footnote{It is easy to generate examples in which this inequality is satisfied. For example, this inequality holds if the signals are independent conditional on the state. } 
\begin{enumerate}
   \item If $\underline l \le \underline \beta < \bar l$, then information does not aggregate if $\jcdf^1$ has conditionally independent signals (denoted by $\cind$) as long as $\cost > 0$; and
     $\exists \cost >0$ and a signal $\jcdf^1 \cad \cind$, such that information aggregates under $\jcdf^1$. 
     \item If $\underline \beta > \bar l$, then information does not aggregate for any $\jcdf^1 \cad \cind$ and any $\cost > 0$. 
\end{enumerate}
\eprop 

\bprf[Proof of Proposition~\ref{Proposition: information aggregation with similarity}]
Substituting the expression for Poisson pdf, we get 
\begin{align*}
    \beta(k, \jcdf,\bar P) =& \frac{\prior}{1-\prior } \frac{\jcdf^1(\bar P,\bar P) \ppdf(k,2\pop) + 2\jcdf^1(\bar P, \bar {NP}) \ppdf(k,\pop)}{\jcdf^0(\bar P,\bar P) \ppdf(k,2\pop) + 2\jcdf^0(\bar P, \bar {NP}) \ppdf(k,\pop)}\\
    =& \frac{\prior}{1-\prior } \frac{\mcdf^1(\bar P) + \jcdf^1(\bar P, \bar P) (e^{-\pop} 2^{k-1} - 1)}{\mcdf^0(\bar P) + \jcdf^0(\bar P, \bar P) (e^{-\pop} 2^{k-1} - 1)}
\end{align*}
Since $2^{k-1}$ is increasing in $k$, 
\begin{align*}
    \sign\left(\frac{\partial \beta(k,\jcdf,\bar P)}{\partial k}\right) = \sign(\jcdf^1(\bar P,\bar P) \mcdf^0(\bar P) - \mcdf^1(\bar P) \jcdf^0(\bar P,\bar P)).
\end{align*}
When signals are conditionally independent in state $1$, $\jcdf^1(\bar P,\bar P) = \mcdf^1(\bar P)^2$. Moreover, for any $\jcdf^1 \cad \cind$, $\jcdf^1(\bar P,\bar P) \ge \mcdf^1(\bar P)^2$. This implies
\begin{align*}
    \sign\left(\frac{\partial \beta(k,\jcdf,\bar P)}{\partial k}\right)= \sign(\mcdf^1(\bar P) \mcdf^0(\bar P) - \jcdf^0(\bar P,\bar P)) > 0.
\end{align*}
The last inequality is true by hypothesis. Therefore, for any $\jcdf$ such that $\jcdf^1 \cad \cind$, $\beta(k,\jcdf,\bar P)$ is increasing in $k$. Moreover, \begin{align*}
    \lim_{k \to \infty} \beta(k,\jcdf,\bar P) = \frac{\prior}{1-\prior}\frac{\jcdf^1(\bar P,\bar P)}{\jcdf^0(\bar P,\bar P)}
\end{align*}
When $\jcdf^1 = \cind$, 
\begin{align*}
    \lim_{k \to \infty} \beta(k,\jcdf,\bar P) = \underline l < \underline \beta
\end{align*}
Therefore, $\beta(k,\jcdf,\bar P) < \underline \beta$ for all $k$ whenever $\jcdf^1 = \cind$. Therefore, $\ind_{\bar P} \notin \eq(\jcdf)$ if $\jcdf^1 = \cind$ as long as $\cost > 0$. That is, information aggregation fails when signals are conditionally independent. 

In contrast, when $\jcdf^1 = \fullcor$, where $\fullcor$ means signals being fully correlated in state $1$, $$\lim_{k \to \infty} \beta(k,\jcdf,\bar P) = \bar l. $$

If $\bar l > \underline \beta$, $\exists k^*$ such that $\beta(k, \jcdf,\bar P) > \underline \beta$ for all $k \ge k^*$. For any $\jcdf^1 \cad \cind$, by Lemma~\ref{Lemma: CAD means transpose of size alpha}, $\jcdf^1(\bar P,\bar P) = \cind(\bar P,\bar P) + \a$, for some $\a \ge 0$. We know that, when $\a = 0$, $\beta(k,\jcdf,\bar P) < \underline \beta$ for all $k$, and, when $\a = \mcdf^1(\bar P) - \mcdf^1(\bar P)^2$, $\exists k^* \in \mathbb N$ such that $\beta(k,\jcdf,\bar P) > \underline \beta$ whenever $k \ge k^* $. Therefore, we can choose an $\a > 0$ small enough to construct $\jcdf^1$ so that $\jcdf^1(\bar P,\bar P) = \cind(\bar P,\bar P) + \a$ and $\beta(k,\jcdf,\bar P) \ge \underline \beta$ if and only if $k > k^*> 2 \pop$. Therefore, the policymaker would use a threshold of $\threshold = k^* > 2 \pop$ when $\jcdf^1$ constructed using $\a$ described above. It is easy to check that $\ppdf(\threshold, 2\pop) > \ppdf(\threshold,\pop)$ in this case. Claim \ref{Claim: when encouragement and P bar cost between IC-P and IC-NP} then establishes that, 
\begin{align*}
    \min_{\signal \in \bar P} \posterior(\signal) \left[ \jcdf^1_\signal(\bar P) \ppdf(\threshold, 2\pop) + \jcdf^1_\signal(\signalset \backslash \bar P) \ppdf(\threshold,\pop) \right] > \max_{\signal \in \signalset\backslash \bar P} \posterior(\signal) \jcdf^1_\signal(\bar P) \ppdf(\threshold, \pop)
\end{align*}
Therefore, by letting $\cost$ to be strictly between the LHS and the RHS of the above, we get that $\ind_{\bar P}$ is an equilibrium, for it satisfies \eqref{Equation: IC P} for all $\signal \in \bar P$ and \eqref{Equation: IC NP} for all $\signal \in \signalset \backslash \bar P$. Therefore, information aggregates under $\jcdf$ wherein $\jcdf^1 \cad \cind$. 

Finally, if $\underline \beta > \bar l$, then the policymaker would not change the status quo regardless of the turnout for any $\jcdf^1 \cad \cind$ establishing the last part of the Proposition. 
 
\bclaim\label{Claim: when encouragement and P bar cost between IC-P and IC-NP}If $\encop$, then 
\begin{align*}
    \min_{\signal \in \bar P} \posterior(\signal) \left[ \jcdf^1_\signal(\bar P) \ppdf(\threshold, 2\pop) + \jcdf^1_\signal(\signalset \backslash \bar P) \ppdf(\threshold,\pop) \right] > \max_{\signal \in \signalset\backslash \bar P} \posterior(\signal) \jcdf^1_\signal(\bar P) \ppdf(\threshold, \pop)
\end{align*}
for any $\jcdf^1 \cad \cind$. 
\eclaim 

\bprf 
First, by definition of $\bar P$, $\min_{\signal \in \bar P} \posterior(\signal) > \max_{\signal \in \signalset \backslash \bar P} \posterior(\signal)$. Also, $\cind_\signal(\bar P) = \mcdf^1(\bar P)$ is independent of $\signal$, and, $\jcdf^1_\signal(\bar P) \ge \mcdf^1(\bar P)$ for all $\signal \in \bar P$ and $\jcdf^1_\signal(\bar P) \le \mcdf^1(\bar P)$ for all $\signal \notin \bar P$. Therefore, $\min_{\signal \in \bar P} \jcdf^1_\signal(\bar P) \ge \max_{\signal \in \signalset \backslash \bar P} \jcdf^1_\signal(\bar P)$. The claim, then, follows due to $\encop$. 

\eprf 
\eprf 

Given an information structure $\jcdf$, we say an equilibrium $\strategy^*$ has ``\emph{maximally informative turnout''} if $I(\strategy^*) \ge I(\strategy)$ for all $\strategy \in \eq(\jcdf)$. We denote $I(\strategy^*)$ by $\infoness(\jcdf)$. Proposition \ref{Proposition: informativeness of equilibrium} below, shows that increasing similarity of information can reduce the informativeness of turnout.

\bprop\label{Proposition: informativeness of equilibrium}
Suppose that $\jcdf^1 \cad \jcdfhat^1$ and $\jcdf^0 = \jcdfhat^0$. Let $\threshold^*$ be the optimal participation threshold for the equilibrium with maximally informative turnout under $\jcdfhat$.
\begin{enumerate}
    \item If 
    $\ppdf(\threshold^*,2\pop) > 2 \ppdf(\threshold^*,\pop)$, then $\infoness(\jcdf) \ge \infoness(\jcdfhat)$ if $\max_{T \subset \signalset} \jcdf^1(T,T) - \jcdfhat^1(T,T)$ is sufficiently small. 
 
 \item If $\ppdf(\threshold^*,2\pop) < \ppdf(\threshold^*,\pop)$, then it is possible that $\infoness(\jcdf) < \infoness(\jcdfhat)$.
\end{enumerate}
\eprop

\bprf[Proof of Proposition \ref{Proposition: informativeness of equilibrium}]
Let $\strategy^*$ be the maximally informative equilibrium under $\jcdfhat$. If the policymaker continues to use $\threshold^*$ as the cutoff, then $\strategy^*$ continues to remain an equilibrium under $\jcdf$ due to Theorem~\ref{Theorem: CAD helps in encouragement}. While this takes care of the incentives of the participants, unlike the earlier arguments, we also need to ensure that a cutoff of $\threshold^*$ is indeed a best response for the policymaker. Since $\threshold^*$ is the cutoff for the maximally informative equilibrium $\beta(\threshold^*;\jcdfhat) \ge \underline \beta$ and $\beta(k;\jcdfhat) < \underline \beta$ for all $k < \threshold^*$. By Bayes' rule, we have, 
\begin{align*}
    \frac{\beta(k;\jcdfhat)}{1-\beta(k;\jcdfhat)}= \frac{\prior}{1-\prior} \frac{\jcdfhat^1(P,P) \ppdf(k,2\pop) + 2\jcdfhat^1(P,NP) \ppdf(k,\pop)}{\jcdfhat^0(P,P) \ppdf(k,2\pop) + 2\jcdfhat^0(P,NP) \ppdf(k,\pop)}
\end{align*}
Since $\ppdf(\threshold^*,2\pop) > 2\ppdf(\threshold^*,\pop)$, $\jcdf^1(P,P) = \jcdfhat^1(P,P) + \a$ and $\jcdf^1(P,NP) = \jcdfhat^1(P,NP) - \a$ for some $\a > 0$ by Lemma~\ref{Lemma: CAD means transpose of size alpha}, and $\jcdf^0 = \jcdfhat^0$, $\beta(\threshold^*;\jcdf) > \beta(\threshold^*;\jcdfhat) \ge \underline \beta$. However, it is now also possible that $\beta(k;\jcdf) \ge \underline \beta$ for some $k < \threshold^8$. Simply lowering the threshold in this case is not an option either as it affects the incentives of the agents, possibly destroying $\strategy^*$ as an equilibrium. However, when $\max_{T\subseteq \signalset} \jcdf^(T,T) - \jcdfhat^1(T,T)$ is sufficiently small, $\beta(k;\jcdf) < \underline \beta$. Finally, since $\infoness(\cdot)$ only depends on the marginal distributions, we obtain the desired inequality. 

For the second part, suppose that $\pop = 20$, $\signalset = \{0,1\}$, $\cost =     0.0368,$ and $\underline \beta = 0.7281$. Signals are conditionally independent in state $0$ with the marginal distribution $\mcdf^0(1) = 0.3$. In state $1$, $\jcdfhat^1(1,1) = 0.66, \jcdfhat^1(1,0) = 0.15$. $\jcdf^1$ is constructed from $\jcdfhat^1$ by using $\a = 0.05$. It is easy to see that $\threshold^*(\jcdfhat) = 28$, while the same no longer constitutes an equilibrium under $\jcdf$. In this case, if an informative equilibrium exists, it must involve mixing. It is easy to check that mixing can only happen on signal $1$, and agents continue to not participate when they receive a signal $0$. Therefore, informativeness under $\jcdf$ is strictly lower than under $\jcdfhat$. 

\eprf

\subsection{On optimal information similarity}\label{Section: optimal information structure}

Consider two extreme cases: conditionally independent signals and perfectly correlated signals. Suppose $Y=(Y_1,Y_2)$ where $Y_i$ is distributed according to $\mcdf^1$, and $Y_1,Y_2$ are independent. Denote this joint distribution by $\cind$. Analogously, let $Y=(Y_1,Y_2)$ be a random variable such that $Y_1 = Y_2$ a.s., and $Y_i$ is distributed according to $\mcdf^1$. We denote this joint distribution by  $\fullcor$. 
Given a conditionally independent signal distribution $\cind$, define $$CI^\uparrow:=\{\ydist \in \De(\signalset \times \signalset): \ydist\cad \cind\}$$
as all the signal distributions that are more similar (in the CAD sense) than $\cind$. Recall that by definition of CAD, all such distributions have the same marginal, and in this case, $\jcdf^0$ does not affect $\protsize(\cdot)$. Therefore, the designer solves the following problem: 
\begin{align*}
    \sup_{\jcdf^1 \in CI^\uparrow } \protsize^*(\jcdf^1,\jcdf^0)
\end{align*}

\bprop[Optimal information similarity]
\label{Proposition: optimal information structure}
An optimal information structure exists. In encouragement environments, fully correlated signals are optimal. In discouragement environments, if conditionally independent signals satisfy Condition M, then they are optimal. 
In other cases, intermediate levels of similarity can be optimal. 
\eprop

\bprf

We prove this using three steps. Steps 1 and 2 establish the existence of an optimal information structure, while Step 3 describes it.

\textbf{Step 1:} We show that $CI^\uparrow$ is weak-$^*$ compact. Consider a sequence $\{\ydist_m\}$ from $CI^{\uparrow}$ that converges to $\ydist$ in the sense that for all $f \in C(\signalset\times\signalset)$, $\int f\dd \ydist_m \to \int f\dd \ydist$. Consider a symmetric $\a \in \real_+^{ \signalset \times \signalset}$, i.e., $\a(i,j) = \a(j,i)$ for all $i,j$, and $\yhatdist \in \De(\signalset \times \signalset)$. Define,
\begin{align*}
    \ydist(i,j) =\yhatdist(i,j) -\a(i,j)\ind_{i\neq j} + \sum_{k \neq i} \a(i,k) \ind_{i=j}
\end{align*}

If $\ydist \in \De(\signalset \times \signalset)$, then we say that ``$\ydist$ is obtained from $\yhatdist$ by an ETI given by $\a$'', denoted by $\ydist = \yhatdist  \biguplus \a$. Recall from the proof of Lemma~\ref{Lemma: CAD means transpose of size alpha} that an alternative characterization of the CAD order (from Proposition 1 in \cite{meyer1990interdependence}) is 
\begin{align*}
    \ydist \cad \yhatdist \Longleftrightarrow \exists \a \in \real_+^{ \signalset\times \signalset} \text{ such that } \ydist = \yhatdist \biguplus \a. 
\end{align*} 
Since $\ydist_m \cad \cind$, we have a sequence $(\a_m) \in \real_+^{\signalset\times\signalset}$ such that $\ydist_m = \cind \biguplus \a_m$. Due to finiteness of $\vert \signalset \times \signalset\vert$, and boundedness of $\a_m$, $\{\a_m\}$ has a convergent subsequence, $\{\a_{m_k}\}$. Let $\a$ be a limit of one such convergent subsequence. Let $\tilde \ydist := \cind \biguplus \a$. Suppose, for contradiction, that $\tilde\ydist\neq \ydist$. Then, there exists some $(x,y) \in \signalset \times \signalset$ such that $\tilde\ydist(x,y) \neq \ydist(x,y)$. Consider a continuous function $f \in C(\signalset\times\signalset)$ such that $f(x',y') =1$ if $x'=x$ and $y'=y$, and $f(x',y')=0$, otherwise.\footnote{Such a function obviously exists because of the finiteness of $\signalset$.} By construction, $\int_{\signalset\times\signalset} f \dd \ydist_{m_k} \to \int_{\signalset\times\signalset} f \dd \tilde\ydist \neq \int_{\signalset\times\signalset} f \dd \ydist$. Hence, a contradiction. Therefore, $\tilde \ydist = \ydist = \cind \biguplus \a$, i.e., $\ydist \cad \cind$, and therefore, $\ydist \in CI^\uparrow$. Therefore, $CI^\uparrow$ is closed.  Finally, since $\signalset \times \signalset$ is compact, $\De(\signalset \times \signalset)$ is weak$-^*$ compact. This makes $CI^\uparrow$ weak$-^*$ compact. 

\textbf{Step 2:} Given compactness, consider a sequence $\{\ydist_m\} \in CI^\uparrow$ such that $\protsize(\ydist_m) \to \bar \protsize=\sup_{\ydist \in CI^\uparrow} \mprotsize(\ydist)$. By compactness, we can (wlog) assume $\ydist_m$ converges to some $\ydist\in CI^\uparrow$. For each $\ydist_m$, let $\strategy^*_m$ be a maximal turnout equilibrium, with participation and not-participation sets~$(P^*_m,NP^*_m)$. Due to the finiteness of $\signalset$, $\vert 2^\signalset\vert$ is finite, and therefore, wlog, $P^*_m = P^*$ for a sufficiently large $m$.\footnote{To be precise, it may be necessary to pass onto a subsequence for this to be true.} Since the marginals are unchanged, we have  $\protsize(\ydist_m)  = 2\pop \mcdf^1(P^*)$ for a sufficiently large $m$. Therefore, $\protsize(\ydist_m) = \bar \protsize$ for a sufficiently large $m$.

\textbf{Step 3:} Since more similarity increases maximal equilibrium turnout in encouragement environments (when $\enco$), we have 
$\mprotsize(\fullcor, \jcdf^0) \ge \mprotsize(\jcdf^1,\jcdf^0)$ for any $\jcdf^1 \in CI^\uparrow$. In discouragement environments (when $\disco$), 
since $\cind$ satisfies Condition M, any $\jcdf^1 \cad \cind$ has $\mprotsize(\jcdf^1,\jcdf^0) \le \mprotsize(\cind,\jcdf^0)$. To see that intermediate levels of similarity can be optimal otherwise, consider the example below. 
\begin{example}
$\popr_1,\popr_2$ follow a Poisson distribution with mean $\pop =15$ and deterministic threshold $\threshold = 20$. The cost of participation is $c= 0.009$ and the signal structure is as follows: $\signalset = \{1,2,3\}$, $\mu_0 = \frac12 $, $\mcdf^1 = [0.25, 0.3, 0.45], \mcdf^0 = [0.6, 0.35, 0.05]$. 
\end{example}

With $\cind$, the unique equilibrium is $\strategy = \ind_{s = 3}$. However, if we perform ETI on the square $\{(1,1),(1,2),(2,1),(2,2)\}$ with $\alpha = 0.005$, we can support an equilibrium $\strategy = \ind_{s\in\{2,3\}}$. Finally, it is easy to see that $\fullcor$ also cannot support $\ind_{s\in\{2,3\}}$ as an equilibrium. We can verify that $\posterior = [0.2941, 0.4615, 0.9000]$. When $P= \{2,3\}$, $\mcdf^1(P) = 0.75 = \cind_\signal(P)$ for all $\signal$ since signals are conditionally independent. $\ppdf(\threshold,\pop) = 0.0418, \ppdf(\threshold,2\pop) = 0.0134$. 
 \eprf

\subsection{Changing information similarity along with changing marginals}\label{Appendix: section on changing information similarity}

Throughout, we have assumed that the agents' marginal distribution of signals in any given state remains unchanged and it is only the joint distribution that changes. However, one may wish to understand the effect of changing information similarity without holding the marginal distribution the same. To this end, it is useful to think of signals as posteriors, as is standard following \cite{Kamenica2011}. That is, we assume that $\signalset$ is the set of posteriors with the natural order, and $\jcdf^1, \jcdf^0$ are two feasible joint distributions over posteriors. As before, we assume that $\signalset$ is finite. Similar to Definition~\ref{Definition: CAD}, we say that $Y \cad \widehat Y$ if only part (2) of that definition holds. That is, we require that, given own signal, an agent assigns a higher probability to the other agent receiving the signal, but we dispense away the requirement of the two signals having the same marginal distributions. This enables us to capture situations such as ones where the two groups exchange each other's information with some probability. We make the following additional assumption that maintain in this section. 
\bass\label{Assumption: increasing for cutoff sets for the information sharing section}Define $g(y,x):= \jcdf^1(\sigr_2 \ge y \vert \sigr_1 =x)$. For every $y$, $g(y,x_1) \ge g(y,x_2)$ if $y \le x_2 \le x_1$.  \eass

\blemma\label{Lemma: optimal equilibrium cutoff in the encouragement environment}Suppose that $\enco$ and $\jcdf^1$ satisfies Assumption~\ref{Assumption: increasing for cutoff sets for the information sharing section}. Then, the maximal equilibrium $\strategy^*$ under $\jcdf$ has a cutoff structure, i.e., $P(\strategy^*) = \signal^\uparrow$ for some some $\signal \in \signalset$.\footnote{$\signal^\uparrow$ is the upper contour set, i.e., the set of signal realizations greater than or equal to $x$.} \elemma

\bprf The proof follows similar steps as the proof of Lemma~\ref{Lemma: restrict to symmetric strategies w/o loss in enco}. Let $\strategy$ be an equilibrium. Define $\underline x$ to be the minimum element of $P(\strategy)$, where $\signalset$ is endowed with the natural order of posteriors. Then, for every $\signal \ge \underline x$, 

\begin{align*}
    \posterior(\signal) \left[\jcdf^1_\signal(\underline x^\uparrow) \pivtwo + (1-\jcdf^1_\signal(\underline x^\uparrow)) \pivone\right]
    & \ge \posterior(\signal) \left[\jcdf^1_{\underline \signal}(\underline x^\uparrow) \pivtwo + (1-\jcdf^1_{\underline \signal}(\underline x^\uparrow)) \pivone\right]\\
    & \ge \posterior(\underline x) \left[\jcdf^1_{\underline \signal}(\underline x^\uparrow) \pivtwo + (1-\jcdf^1_{\underline \signal}(\underline x^\uparrow)) \pivone\right] \ge \cost 
\end{align*}
where the first inequality follows from Assumption~\ref{Assumption: increasing for cutoff sets for the information sharing section} and the second inequality is due to $\posterior(\signal) \ge \posterior(\underline x)$ whenever $\signal \ge \underline x$. 

Therefore, \eqref{Equation: IC P} is satisfied for all $\signal \ge \underline x$. Hence, following identical steps as in Lemma~\ref{Lemma: restrict to symmetric strategies w/o loss in enco} with $P_1 = P_2 = \underline \signal^\uparrow$, we obtain that the maximal equilibrium is in symmetric, cut-off strategies. 
\eprf 

For the proposition below, we let $\jcdf, \jcdfhat$ be two joint distributions and let $\margdist^\state, \widehat\margdist^\state$ be the joint distribution of signals (and hence posteriors) they induce in state $\state$. Following \cite{kuvalekar2023wrong}, we now define what it means for a random variable to be ``more spread out'' than another. Similar notions have appeared in \cite{johnson2006simple}. 
\bdefn\label{Definition: More spread out} We say that a random variable $Y$ is ``more spread out around $y$'' than $\widehat Y$ if, $G_Y(z) \ge G_{\widehat Y}(z)$ for all $z \le y$ and $G_Y(z) \le G_{\widehat Y}(z)$ for all $z \ge y$. \edefn

\bprop\label{Proposition: changing marginal result}Suppose that $\enco$, $\jcdf^1 \cad \jcdfhat^1$, and $\margdist^1$ is more spread out around $\tilde \posterior$ than $\widehat\margdist^1$ for some $\tilde \posterior \le \cost$. Then, $\protsize^*(\jcdf) \ge \protsize^*(\jcdfhat)$. 
\eprop 

\bprf 
By Theorem~\ref{Theorem: CAD helps in encouragement}, if $\strategy \in \eq(\jcdfhat)$, then $\strategy \in \eq(\jcdf)$ whenever $\enco$. Notice that the proof of this part in Theorem~\ref{Theorem: CAD helps in encouragement} did not use the fact that the marginal distributions remained unchanged therein. By Lemma~\ref{Lemma: optimal equilibrium cutoff in the encouragement environment}, the maximal equilibrium is in cutoff strategies for $\jcdf$ and $\jcdfhat$. Let $\strategy(\jcdfhat)$ denote the maximal equilibrium under $\jcdfhat$. Since it has a cutoff structure, let $\widehat \signal$ be the associated cutoff. An agent would never participate on a signal such that $\posterior(\signal) < \cost$. Therefore, $ \posterior(\widehat\signal) \ge \cost$. Hence, \begin{align*}
    \protsize^*(\jcdfhat) = \widehat \margdist^1(\widehat \signal^\uparrow)\le \margdist^1(\widehat\signal^\uparrow)
\end{align*}
where the inequality follows due to the spread-out ranking of $\margdist^1$ and $\widehat \margdist^1$. Thus, $\protsize^*(\jcdf) \ge \protsize^*(\jcdfhat)$. 
\eprf 
Finally, unlike in Theorem~\ref{Theorem: CAD helps in encouragement}, more information can strictly lower participation even in the encouragement environment without an additional condition such as the spread-out order.\footnote{An example is available on request.} 

\subsection{Example from the introduction}\label{Appendix: intro example}
In this section, we analyze the example from the introduction in more detail. Recall the example. The payoff matrix is as follows. 
\begin{center}
\begin{table}[ht]
    \centering
    \begin{tabular}{cc|c|c|}
            & \multicolumn{1}{c}{} &                     \multicolumn{2}{c}{Bob}\\
            & \multicolumn{1}{c}{} & \multicolumn{1}{c}{work} & \multicolumn{1}{c}{shirk} \\ 
            \cline{3-4}
   \multirow{2}{*}{Abe}     & work &$\theta-c,\theta -c$& $q\theta -c, q \theta$ \\ 
            \cline{3-4}
    & shirk & $q\theta,q\theta - c$ &$0,0$ \\ 
   \cline{3-4}
\end{tabular}
\end{table}
    \end{center}
\vspace*{-6mm}
Abe and Bob each receives a binary signal $\sigr_i\in \{0,1\}$, such that  $\sigr_i = 0$ if $\stater = 0$, and if $\stater = 1$, the signals are drawn from some exchangeable joint distribution $\jcdf^1$. Table~\ref{tab:intro example} below describes~$\jcdf^1$. 

\begin{table}[h]
    \centering
    \begin{tabular}{c|c c|c}
        \toprule
        $\stater=1$ & $\sigr_{-i}=0$ & $\sigr_{-i}=1$ & \text{marg} \\
        \midrule
        $\sigr_{i}=0$ & $(1-p)^2+\alpha$ & $p(1-p)-\alpha$ & $1-p$\\
        $\sigr_{i}=1$ & $p(1-p)-\alpha$ & $p^2+\alpha$ & $p$ \\
        \midrule 
        \text{marg} & $1-p$ & $p$ & $1$ \\
        \bottomrule
    \end{tabular}
    \caption{Probability distribution of signals when $\stater =1$}
    \label{tab:intro example}
\end{table}
Notice that the marginal in state~$\stater=1$ is given by ~$P(\sigr_i=1|\stater=1)=p$, and $\posterior(0)\in(\frac{1}{2},1)$. 
Higher $\alpha$ makes Abe and Bob's signals more similar conditional on the state being~$\theta=1$: Signals are independent when $\alpha=0$ and perfectly correlated when $\alpha=p(1-p)$. Define 
$$\tilde p:=\jcdf^1_1(\sigr_{-i}=1)=p+\frac{\alpha}{p}.$$ 
Obviously, $\tilde p\geq p$ and the equality holds under conditional independence. Given the primitive $p$, we can say that the signals are more similar when $\tilde p$ is higher. 

For simplicity, we assumed in the example that a player never works after seeing $\sigr_i=0$, or formally,
  $$\frac{c}{\mu(0)}>1.$$
This assumption holds when $c>\frac{1-p}{1+1-p}$, or $p>\frac{1-2c}{1-c}$. If $c>\frac{1}{2}$, then this assumption always holds true, and when $c<\frac{1}{2}$, this assumption requires $p$ to be sufficiently high. For simplicity, we assume
$$c>\frac{1}{2}.$$

Let us start with a numerical example. 
Suppose that $\cost=0.6$ and $Prob(\sigr_i = 1 \vert \state = 1) = \frac13$. 

Consider signals that are independent conditional on $\stater=1$. Then, 
$\jcdf^1_1(\sigr_{-i}=1)=\frac13$. If $\jcdf^1_1(\sigr_{-i}=1)<\frac12$, a player's incentive to work is more influenced by whether she alone can make a difference, making the LHS of the incentive constraint increasing in $q$. Therefore, if $\strategy^1$ is an equilibrium for some $q$, it is an equilibrium for any $q'>q$. This gives us the first result with  $q^*=0.8$. 

Next, we increase the similarity of the signal as in Figure~\ref{Figure: CAD increase, binary, intro}, where $\alpha=\frac29$. This makes $\jcdf^1_1(\sigr_{-i}=1) =1$. If $\jcdf^1_1(\sigr_{-i}=1)>\frac12$, a player's incentive to work is more influenced by whether the other player works, making the LHS on the incentive constraint decreasing in $q$. Therefore, if $\strategy^1$ is an equilibrium for some $q$, it is an equilibrium for all $q'<q$. This gives us the second result with $q^{**}=0.4$.

Below, we show that the qualitative insight demonstrated in the introduction is valid even if we allowed asymmetric strategies or mixing. For any $\beta\in[0,1]$, define $I^\beta(\tilde p,q)$ to be the probability of being pivotal when working, given that the other player works with probability $\b$ when he receives a signal~$\sigr_{-i}=1$. $$I^1(\tilde p,q):=\tilde p (1-q)+(1-\tilde p)q.$$
$$I^0(\tilde p,q):=q.$$
For any $\beta\in(0,1)$, $$I^\beta(\tilde p,q):=\tilde p (\beta(1-q)+(1-\beta)q)+(1-\tilde p)q.$$
Note that $I^\beta(\tilde p,q)$ is decreasing in $\beta$ if $q>\frac{1}{2}$ and increasing in $\beta$ if $q<\frac{1}{2}$. If $I^0(\tilde p,q),I^1(\tilde p,q)<c$, then $I^\beta(\tilde p,q)<c$ for all $\beta$.  

In the introduction, we only considered the symmetric pure strategy profile~$\sigma^1$ in which both Abe and Bob work with probability $1$ after seeing $\sigr_{i}=1$. Now let~$\sigma^\beta$ denote symmetric mixed strategy profiles in which both players work with probability $\beta$ after seeing $\sigr_{i}=1$. Let $\sigma^a$ denote an asymmetric strategy profile in which exactly one of Abe and Bob works after seeing $\sigr_{i}=1$, and let ~$\sigma^0$ denote the strategy profile with no effort even after $\sigr_i=1$. We say an equilibrium is maximal if it maximizes the total expected effort by the two players. 

\begin{proposition}
    Assume $c>\frac{1}{2}$. 
    \begin{enumerate}
        \item If $q<\frac{1}{2}$ then 
        \begin{enumerate}
            \item for $\tilde p<c$, $\sigma^0$ is the maximal equilibrium 
            \item for $\tilde p>c$, there exists a threshold $q^{**}$ such that, $\sigma^1$ is the maximal equilibrium whenever $q\in(0,q^{**})$ and $\sigma^0$ is the maximal equilibrium in $q\in(q^{**},\frac{1}{2})$
        \end{enumerate}
        \item[] The threshold $q^{**}$ is increasing in $\tilde p$. Therefore, the region where $\strategy^1$ is the maximal equilibrium is increasing in the set order in $\tilde p$. Therefore, similarity helps participation.
         \item If $q>\frac{1}{2}$, then there exist two thresholds, $\widehat q, q^*$ such that the following holds:
        \begin{enumerate}
            \item for $\tilde p<1-c$,  $\sigma^0$ is the maximal equilibrium in $q\in(\frac{1}{2},c)$,  $\sigma^a$ is the maximal equilibrium in $q\in(c,\widehat q)$,  $\sigma^\beta$ is the maximal equilibrium in $q\in(\widehat q, q^*)$,  $\sigma^1$ is the maximal equilibrium in $q\in(q^*,1)$
            \item for $\tilde p>1-c$, $\sigma^0$ is the maximal equilibrium in $q\in(\frac{1}{2},c)$,  $\sigma^a$ is the maximal equilibrium in $q\in(c,\widehat q)$,  $\sigma^\beta$ is the maximal equilibrium in $q\in(\widehat q, 1)$.
        \end{enumerate}
        \item[]  The thresholds $\widehat q,q^*$ are increasing in $\tilde p$, and the mixing probability $\beta$ is decreasing in $\tilde p$. Therefore, similarity hurts participation.
    \end{enumerate}

\end{proposition}

\begin{proof}

Define 
\begin{align*}
  q^*(\tilde p):=&\max\left\{\min\left\{\frac{c-\tilde p}{1-2\tilde p},1\right\},0\right\}  \\
  q^{**}(\tilde p):=&\max \{ \min \left\{\frac{\tilde p-c}{2\tilde p-1},1\right\},0\}\\
  \widehat q (\tilde p):=&\max \left\{ \min \left\{\frac{1}{2}\left( \frac{2c-1}{1-\tilde p}+1\right),1\right\},0\right\}.
\end{align*}

All three functions are non-decreasing in $\tilde p$. Since $q^*(0)=\widehat q(0)=c$, for any $\tilde p$, $q^*(\tilde p),\widehat q(\tilde p) \geq c$. Recall that $c>\frac{1}{2}$. On the other hand, for any $\tilde p$, $q^{**}(\tilde p)\leq q^{**}(1)=1-c<\frac{1}{2}$. Also, note that $q^*(1-c)=1$ and $q^{**}(c)=0$. 

\bclaim
\label{claim:sigma1}
For any $\tilde p\in [p,1]$, $\sigma^1$ is the maximal equilibrium if $q\geq q^*(\tilde p)$ or if $q\leq q^{**}(\tilde p)$. Moreover, if $\tilde p\in(1-c,c)$, then $q^*(\tilde p)=1$ and $q^{**}(\tilde p)=0$, and $\sigma^1$ is not an equilibrium for any $q\in[0,1]$. 
\eclaim
\bprf
$\sigma^1$ is an equilibrium if $I^1(\tilde p,q)\geq c$. If $\sigma^1$ is an equilibrium, then it is the maximal equilibrium. Suppose that $\tilde p\leq \frac{1}{2}$. Then, $I^1(\tilde p,q)$ is increasing in $q$. Therefore, there exists $q^*(\tilde p)$ (defined above)
such that for $q\geq q^*(\tilde p)$, $I^1(\tilde p,q)\geq c$, i.e, $\sigma^1$ is the maximal equilibrium. Next, suppose that $\tilde p\geq \frac{1}{2}$. Then, $I^1(\tilde p,q)$ is decreasing in $q$. Therefore, there exists $q^{**}(\tilde p)$ (defined above) such that for $q\leq q^{**}(\tilde p)$, $I^1(\tilde p,q)\geq c$, i.e, $\sigma^1$ is the maximal equilibrium. 

If $\tilde p<1-c$, then $q^{**}(\tilde p)=0$ and $q^*(\tilde p)\in (c,1)$. Therefore, $\sigma^1$ is not an equilibrium for $q<\frac{1}{2}$, whereas it is an equilibrium only if $q$ is sufficiently high. If $\tilde p=1-c$, then $q^*(\tilde p)=\frac{c-\tilde p}{1-2\tilde p}=1$, and $\sigma^1$ is an equilibrium only if $q=1$. For $\tilde p \in (1-c,c)$, $\sigma^1$ is not an equilibrium for any $q$.  If $\tilde p=c$, then $q^{**}(\tilde p)=\frac{\tilde p-c}{2\tilde p-1}=0$, and $\sigma^1$ is an equilibrium only if $q=0$. If $\tilde p>c$, then $q^{**}(\tilde p)\in(0,1-c)$ and $q^*(\tilde p)=1$. Therefore, $\sigma^1$ is not an equilibrium for $q>\frac{1}{2}$, whereas it is an equilibrium only if $q$ is sufficiently low.

\eprf

\bclaim
\label{claim:sigma0}
For any $\tilde p\in [p,1]$, $\sigma^0$ is the maximal equilibrium if $q\in (q^{**}(\tilde p),c)$. 
\eclaim

\bprf
If $q<c\leq q^*(\tilde p)$ and $q>q^{**}(\tilde p)$, then $\sigma^1$ is not an equilibrium, i.e., $I^1(\tilde p, q)<c$. Moreover $I^0(\tilde p,q)=q<c$, which means $\sigma^a$ cannot be an equilibrium either. Recall that if $I^1(\tilde p, q)<c$ and $I^0(\tilde p,q)<c$, then $I^\b(\tilde p,q)<c$ for any $\b$, which means $\sigma^\b$ cannot be an equilibrium either. Note that $\sigma^0$ is always an equilibrium for $q<c$. Therefore, it is the only equilibrium, and hence, the maximal equilibrium.  
\eprf 

Since $c>\frac{1}{2}$, it follows from the above two claims that for $q<\frac{1}{2}$, the maximal equilibrium is either $\sigma^1$ or $\sigma^0$. Recall that $q^{**}(\tilde p)=0$ for all $\tilde p<c$. Therefore, for $\tilde p<c$ the maximal equilibrium is $\sigma^0$ for all $q\in[0,\frac{1}{2}]$ and for  $\tilde p>c$ the maximal equilibrium is $\sigma^1$ for all $q\in[0,q^{**}(\tilde p)]$ and $\sigma^0$ for all $q\in[q^{**}(\tilde p),\frac{1}{2}]$.  This proves the first part of the proposition. 

\bclaim
For any $\tilde p\in [p,1]$, $\sigma^a$ is the maximal equilibrium if $q\in [c, \widehat q(\tilde p)]$ and $\sigma^\b$ is the maximal equilibrium if $q\in [\widehat q(\tilde p), q^*(\tilde p)]$. 
\eclaim

\bprf 

For any $\tilde p$, $\widehat q (\tilde p)\in [c,q^*(\tilde p)]$. It follows from Claim \ref{claim:sigma1} that $\sigma^1$ is not an equilibrium in the interval $[c,q^*(\tilde p)]$. In this interval, $\sigma^a$ is an equilibrium since $I^0(\tilde p,q)=q\geq c$. In this interval, there is also a mixed strategy equilibrium $\sigma^\b$ where the players are indifferent between exerting effort and not exerting effort after seeing signal $1 $, i.e., $I^\beta(\tilde p,q)=c$. This gives us the equilibrium  probability of effort
$$\b=\frac{q-c}{(2q-1)\tilde p}=\frac{1}{\tilde p} \left( \frac{1}{2}-\frac{c-\frac{1}{2}}{2q-1}\right).$$
 The expected effort under $\sigma^a$ is $p$ whereas that under $\sigma^\b$ is $2\beta p$. Therefore, in this interval, $\sigma^\b$ is the maximal equilibrium if $\beta\geq \frac{1}{2}$ and $\sigma^a$ is the maximal equilibrium, otherwise. $\b \geq \frac{1}{2}$ is equivalent to $q\geq \widehat q(\tilde p)$. Since $\b$ is decreasing in $\tilde p$ and increasing in $q$, $\widehat q(\tilde p)$ increasing in $\tilde p$.

\eprf 

Recall that for any $\tilde p$, $\widehat q(\tilde p),q^*(\tilde p)\geq c>\frac{1}{2}$ and $q^*(\tilde p)=1$ for $\tilde p>1-c$. Therefore, under $\tilde p<1-c$, for $q>\frac{1}{2}$, $\sigma^0$ is the maximal equilibrium in $q\in(\frac{1}{2},c)$,  $\sigma^a$ is the maximal equilibrium in $q\in(c,\widehat q(\tilde p))$,  $\sigma^\beta$ is the maximal equilibrium in $q\in(\widehat q,(\tilde p) q^*(\tilde p))$,  $\sigma^1$ is the maximal equilibrium in $q\in(q^*(\tilde p),1)$. Finally, under $\tilde p>1-c$, $q^*(\tilde p)=1$ and $\sigma^1$ is not an equilibrium at all. This proves the second part of the proposition. 
 
\end{proof}

Figure \ref{fig:binaryexample} provides a visual description of how the maximal equilibrium changes with information similarity. In this simple example, under conditionally independent signals $\tilde p=p<1-c<\frac{1}{2}$, $I^1(\tilde p,q)$ increases in $q$. That is, a player has a higher incentive to work if he is more likely to complete the project on his own. Therefore, there is a $q^*>c$ such that $\sigma^1$ is an equilibrium for $q>q^*$. For 
$q<q^*$, $\sigma^1$ cannot be sustained as an equilibrium. However, if $q>c$, there is an asymmetric and mixed equilibrium. A higher $\tilde p$ means $I^1(\tilde p, q)$ is lower (since $q>\frac{1}{2}$). Thus, similar information reduces a player's incentive to work for high $q$, and it becomes more difficult to sustain $\sigma^1$ as an equilibrium ($q^*$ increases). In fact, if $\tilde p>1-c$, then $q^*=1$, i.e., $\sigma^1$ cannot be an equilibrium for any $q>\frac{1}{2}$. Even for the mixed strategy equilibrium, the probability of working must decrease to sustain an equilibrium.  As $\tilde p$ increases further ($\tilde p>c$), then $I^1(\tilde p,q)>c$ for sufficiently small values of $q$ ($q<q^{**}$). That is, players who are unlikely to complete the project on their own, are now willing to work because it is very likely ($\tilde p>c$), that when they see $1$ the other also sees $1$. As $\tilde p$ increases further, $I^1(\tilde p,q)$ increases, that is, it becomes easier to sustain $\sigma^1$ as an equilibrium for low values of $q$, and accordingly $q^{**}$ increases.

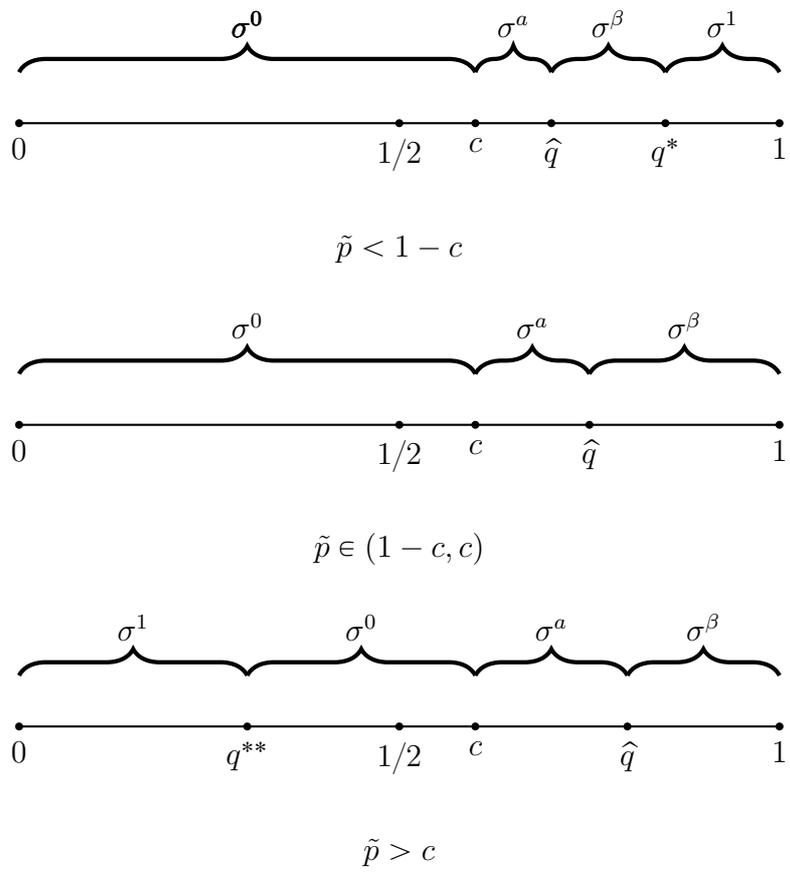
\begin{figure}[h!]
\centering
\begin{tikzpicture}
\draw[thick,-](0,0) to (10,0); 
\node[circle,fill=black,inner sep=0pt,minimum size=3pt,label=below:{$0$}] (0) at (0,0)  {};

\node[circle,fill=black,inner sep=0pt,minimum size=3pt,label=below:{$1$}] (1) at (10,0)  {};

\node[circle,fill=black,inner sep=0pt,minimum size=3pt,label=below:{$1/2$}] (1/2) at (5,0)  {};

\node[circle,fill=black,inner sep=0pt,minimum size=3pt,label=below:{$c$}] (c) at (6,0)  {};

\node[circle,fill=black,inner sep=0pt,minimum size=3pt,label=below:{$\widehat q$}] (qhat) at (7,0)  {};

\node[circle,fill=black,inner sep=0pt,minimum size=3pt,label=below:{$q^*$}] (qstar) at (8.5,0)  {};

\node at (3,1)[above]{$\sigma^0$};

\draw [ultra thick, decorate,
    decoration = {brace,
        raise=5pt,
        amplitude=10pt}] (0,0.5) --  (6,0.5);
        \node at (3,1)[above]{$\sigma^0$};

\node at (6.5,1)[above]{$\sigma^a$};        

\draw [ultra thick, decorate,
    decoration = {brace,
        raise=5pt,
        amplitude=10pt}] (6,0.5) --  (7,0.5);

\node at (7.75,1)[above]{$\sigma^\beta$};        

\draw [ultra thick, decorate,
    decoration = {brace,
        raise=5pt,
        amplitude=10pt}] (7,0.5) --  (8.5,0.5);    

\node at (9.25,1)[above]{$\sigma^1$};        

\draw [ultra thick, decorate,
    decoration = {brace,
        raise=5pt,
        amplitude=10pt}] (8.5,0.5) --  (10,0.5);

\node at (5,-2)[above]{$\tilde p<1-c$};        

\draw[thick,-](0,-4) to (10,-4); 
\node[circle,fill=black,inner sep=0pt,minimum size=3pt,label=below:{$0$}] (0) at (0,-4)  {};

\node[circle,fill=black,inner sep=0pt,minimum size=3pt,label=below:{$1$}] (1) at (10,-4)  {};

\node[circle,fill=black,inner sep=0pt,minimum size=3pt,label=below:{$1/2$}] (1/2) at (5,-4)  {};

\node[circle,fill=black,inner sep=0pt,minimum size=3pt,label=below:{$c$}] (c) at (6,-4)  {};

\node[circle,fill=black,inner sep=0pt,minimum size=3pt,label=below:{$\widehat q$}] (qhat) at (7.5,-4)  {};


\node at (3,-3)[above]{$\sigma^0$};

\draw [ultra thick, decorate,
    decoration = {brace,
        raise=5pt,
        amplitude=10pt}] (0,-3.5) --  (6,-3.5);
        \node at (3,1)[above]{$\sigma^0$};

\node at (6.75,-3)[above]{$\sigma^a$};        

\draw [ultra thick, decorate,
    decoration = {brace,
        raise=5pt,
        amplitude=10pt}] (6,-3.5) --  (7.5,-3.5);

\node at (8.75,-3)[above]{$\sigma^\beta$};        

\draw [ultra thick, decorate,
    decoration = {brace,
        raise=5pt,
        amplitude=10pt}] (7.5,-3.5) --  (10,-3.5);

\node at (5,-6)[above]{$\tilde p\in(1-c,c)$};     

\draw[thick,-](0,-8) to (10,-8); 
\node[circle,fill=black,inner sep=0pt,minimum size=3pt,label=below:{$0$}] (0) at (0,-8)  {};

\node[circle,fill=black,inner sep=0pt,minimum size=3pt,label=below:{$1$}] (1) at (10,-8)  {};

\node[circle,fill=black,inner sep=0pt,minimum size=3pt,label=below:{$1/2$}] (1/2) at (5,-8)  {};

\node[circle,fill=black,inner sep=0pt,minimum size=3pt,label=below:{$c$}] (c) at (6,-8)  {};

\node[circle,fill=black,inner sep=0pt,minimum size=3pt,label=below:{$\widehat q$}] (qhat) at (8,-8)  {};

\node[circle,fill=black,inner sep=0pt,minimum size=3pt,label=below:{$q^{**}$}] (qstarstar) at (3,-8)  {};

\node at (1.5,-7)[above]{$\sigma^1$};

\draw [ultra thick, decorate,
    decoration = {brace,
        raise=5pt,
        amplitude=10pt}] (0,-7.5) --  (3,-7.5);
        \node at (3,1)[above]{$\sigma^0$};

\node at (4.5,-7)[above]{$\sigma^0$};

\draw [ultra thick, decorate,
    decoration = {brace,
        raise=5pt,
        amplitude=10pt}] (3,-7.5) --  (6,-7.5);
        \node at (3,1)[above]{$\sigma^0$};

\node at (7,-7)[above]{$\sigma^a$};        

\draw [ultra thick, decorate,
    decoration = {brace,
        raise=5pt,
        amplitude=10pt}] (6,-7.5) --  (8,-7.5);

\node at (9,-7)[above]{$\sigma^\beta$};        

\draw [ultra thick, decorate,
    decoration = {brace,
        raise=5pt,
        amplitude=10pt}] (8,-7.5) --  (10,-7.5);

\node at (5,-10)[above]{$\tilde p>c$};   
        
\end{tikzpicture}
\caption{Information Similarity and Maximal Equilibrium}
\label{fig:binaryexample}
\end{figure}

\end{document}